\documentclass[fontsize,12pt,times]{elsarticle}
\usepackage[left=2cm,right=2cm,top=3cm,bottom=3cm]{geometry}
\usepackage{amssymb}
\usepackage{amsthm}
\usepackage{latexsym}
\usepackage{amsmath}
\usepackage{xcolor}
\usepackage{graphicx}
\usepackage{indentfirst}
\usepackage{mathrsfs}
\usepackage{pdfsync}
\usepackage{hyperref}
\hypersetup{hypertex=true,
	colorlinks=true,
	linkcolor=blue,
	anchorcolor=blue,
	citecolor=blue}
\allowdisplaybreaks
\usepackage[utf8]{inputenc}

\newtheorem{theorem}{\color{black}\indent Theorem}
\newtheorem{lemma}{\color{black}\indent Lemma}[section]
\newtheorem{proposition}{\color{black}\indent Proposition}
\newtheorem{definition}{\color{black}\indent Definition}[section]
\newtheorem{remark}{\color{black}\indent Remark}[section]
\newtheorem{corollary}{\color{black}\indent Corollary}[section]

\begin{document}
\begin{frontmatter}
\title{Semiclassical scar on tori in high dimension }
\author{{ Huanhuan Yuan$^{a,}$ \footnote{ E-mail address : yuanhh128@nenu.edu.cn}
		,~ Yong Li$^{b,c,}$  \footnote{E-mail address : liyong@jlu.edu.cn}
        \footnote{Mathematics Subject Classification (MSC 2020): 81Q20, 37J40, 70H08.} \\
	{$^{a}$School of Mathematics and Statistics, Northeast Normal University,} {Changchun 130024, P. R. China.}\\
    {$^{b}$School of Mathematics, Jilin University, Changchun 130012, P. R. China.} \\
	{$^{c}$Center for Mathematics and Interdisciplinary Sciences, Northeast Normal University,}
	{Changchun, 130024, P. R. China.}
}}

\begin{abstract}
We show that the eigenfunctions of the self-adjoint elliptic  $h-$differential operator $P_{h}(t)$ exhibits semiclassical scar phenomena on the
$d-$dimensional torus, under the $\sigma$-Bruno-R\"{u}ssmann condition, instead of the Diophantine one. Its equivalence is described as: for almost all perturbed Hamiltonian's KAM Lagrangian tori $\Lambda_{\omega}$, there exists a semiclassical measure with positive mass on $\Lambda_{\omega}$. The premise is that we can obatain a family of quasimodes for the $h-$differential operator $P_{h}(t)$  in the semiclassical limit as $h\rightarrow0$, under the $\sigma$-Bruno-R\"{u}ssmann condition.
\end{abstract}
\begin{keyword}
Eigenvalue, Quasimode,  semiclassical scarring, positive mass, energy window
\end{keyword}
\end{frontmatter}

\section{Introduction}
Semiclassical scars are an intriguing phenomenon in quantum systems, where, despite the system's chaotic behavior at the classical level, the quantum wave function localizes near certain classical periodic orbits. This phenomenon, known as scars, is primarily observed in the semiclassical limit, where the system's energy is low (i.e., the Planck constant
$h$ is small) or the wavelength is long. While the classical system is chaotic, the quantum wave function may still localize near these periodic orbits, leaving behind a scar. This highlights the transition between quantum and classical descriptions, particularly in the low-energy regime, where classical orbits significantly influence the quantum state.

Although the classical chaos, semiclassical methods provide insights into this phenomenon by linking quantum states to classical dynamics, such as WKB theory together with Gutzwiller's trace formula. The understanding of semiclassical scars helps to explore the quantum-classical relationship, offering deeper insights into quantum chaos, dynamics, and information. Moreover, this phenomenon has potential applications in quantum computing, particularly in the control and stabilization of quantum states.

In the physics literature, if an eigenfunction exhibits slight enhancement near unstable periodic orbits, this phenomenon is commonly referred to as a scar by physicists
(Helffer \cite{MR880418}). Accordingly, in the mathematics literature, a sequence of eigenfunctions is said to be strongly scarred if the corresponding semiclassical measure $\mu$ is supported on some periodic trajectory (Anantharaman \cite{MR4477342}).

In reference \cite{MR4404789}, Gomes proved that the eigenfunctions of pseudodifferential operators exhibit the semi-classical scar phenomenon on the 2-dimensional torus of Hamiltonian systems, under the Diophantine non-resonance condition. In the 3-dimensional contact case, Colin de Verdi\`{e}re, Y. et al. \cite{MR3560358} presented the eigenfunctions of any pseudodifferential operator of order 0, which concentrate on $\Sigma\subset T^{*}M$ under the Diophantine condition. In another work, they also studied spectral asymptotics for sub-Riemannian (sR) Laplacians \cite{MR3743700}. In addition, Anantharaman, N. in \cite{MR4477342} showed that strong
scarring is not possible on a negatively curved manifold. However, it does not rule out a partial scar, in other words we cannot exclude that $\mu(\Gamma)> 0$, where $\Gamma$ is a periodic orbit.

\subsection{Introduction to Semiclassical scar}
From \cite{MR2051153}, we know that a quantization of the classical system just described is a semiclassical family of pseudodifferential operators $P_{h}(t)$, depending on a small semiclassical parameter $h\in (0,h_{0}]$ and time parameter $t\in(0,t_{0}]$, with semiclassical principal symbol $p(x,\xi;t)$. Assume that $P_{h}(t)$ has fixed positive differential order $J$ and is an elliptic and self-adjoint operator acting on half-densities in $L^{2}(M;\Omega^{1/2})$, where $L^{2}(M;\Omega^{1/2})$ can be identified with $L^{2}(M)$ by writing $u=u(x)|dx|^{\frac{1}{2}}, u\in L^{2}(M)$.   Moreover, let's assume that the subprincipal symbol of $P_{h}(t)$ is null.

One way to understand the limits $h\rightarrow0$ of a family of functions
$u = {u(h)}_{0<h\leq h_{0}}$, where $u_{h}$ is bounded in $L^{2}$, is to define corresponding semiclassical measures $\nu$. These measures capture the limiting behavior of certain quadratic forms acting on $u(h)$. The details also can read \cite{MR2952218}. We formulate our result in terms of semiclassical measures. Let $\nu$ be a positive measure on $T^{*}M$.
\begin{definition}\label{a12}
We call $\nu$ a semiclassical measure associated with the function sequence $v(h_{j})$ (for a sequence $h_{j}\downarrow0$), if there exists a negative Radon measure $\nu$ such that
\begin{align}\label{bk}
\langle P_{h_{j}}(t)v(h_{j}),v(h_{j})\rangle\rightarrow\int_{T^{*}M} \sigma(P_{h_{j}}(t))d\nu,\ as j\rightarrow \infty
\end{align}
for every semiclassical pseudodifferential operator $P_{h_{j}}(t)$ of semiclassical order 0 and compact microsupport, where $\sigma(P)$ denotes the principal symbol of operator $P$, and $\langle\cdot,\cdot\rangle$ is inner product in  the usual sense.
\end{definition}
These measures describe the possible ways in which the mass of eigenfunctions of $P_{h}(t)$ can distribute in the position-frequency space as $h\rightarrow0$. Semiclassical measures for the {S}chr\"odinger equation on the torus can refer to reference \cite{MR3226742}. Compactness theorem shows that each norm 1 sequence $u_{h_{j}}$ with compact support has a subsequence to a weak limit in the sense of \eqref{bk}. It holds for a sequence of normalized eigenfunctions of $P_{h_{j}}(t)$ with uniformly bounded eigenvalues as $h_{j}\rightarrow0$.  Every semiclassical measure is a probability measure on $\mathbb{T}^{d}$.
\par
Let $(\phi_{n}(t,h))_{n\in\mathbb{N}_{+}}$ ba a (complex-valued) Hilbert basis of $L^{2}(M,\mu)$, consisting of eigenfunctions of $P_{h}(t)$, associated with the eigenvalue $E_{j}(t,h)$ for each $h\in(0,h_{0}]$ and $t\in(0,t_{0}]$. We are interested in the behaviour of these eigenfunctions of $P_{h}(t)$ in the semiclassical limit $h\rightarrow 0$.
\par
Let $E=\{x\in T^{*}M | p(x,t)=0\}$ be a smooth connected compact manifold equipped with an arbitrary non-vanishing smooth density $\mu$. Suppose that $dp(\cdot,t)$ does not vanish on $E$; this implies that E is a smooth codimension-1 submanifold of $T^{*}M$. The Liouville measure $\mu$ on $T^{*}M$ induces a smooth measure $\sigma_{E}$ on $E$ by writing $\mu=\sigma\wedge dp(\cdot,t)$, where $\sigma\wedge dp(\cdot,t)$ can be regarded as the volume element on the manifold, and then restricting $\sigma$ on $E$. We say that the sequence of eigenfunction scars at $S$, if $\nu$ has positive mass on $S\subset E$ of $\sigma_{E}-$ measure zero.

\subsection{Regularity assumptions and Non-resonance conditions}
In this paper, we aim to show that the eigenfunctions of $P_{h}(t)$ exhibit semiclassical scar phenomena in high dimension under $\sigma$-Bruno-R\"{u}ssmann condition (a new non-resonance condition we proposed), instead of the Diophantine one. Here, the eigenfunctions we study correspond to eigenvalues within the energy window, and their maximum values have a lower bound independent of $h$, which is related to the maximum value of the quasimodes (see Definition \ref{a9}). We prove that such eigenfunctions exhibit semiclassical  scar phenomena. It can also be stated that we find a semiclassical measure with positive mass in $\mathbb{T}^{d}$.

In canonical perturbation theory, such as Gallavotti \cite{MR733479}, from the perspective of non-resonant conditions, the Diophantine condition is not necessary. In fact, it is adequate to assume that $\omega$ satisfies $\sigma$-Bruno-R\"{u}ssmann condition, that is,
\begin{align}\label{a}
\Omega_{\sigma}:=\left\{\omega\in\Omega:\langle\omega,k\rangle\geq \frac{\kappa}{\Delta(|k|)},\ k\in \mathbb{Z}^{d}\setminus \{0\},|k|=\sum_{i=1}^{n}|k_{i}|\right\}
\end{align}
with some $\kappa>0$ and some $\sigma-$approximation function $\Delta(t)$, which is a continuous, strictly increasing, unbounded function and satisfies
\begin{align}\label{i1}
\frac{\log \Delta(t)}{t^{\frac{1}{\sigma}}}\searrow 0,\ 0<t\rightarrow0,
\end{align}
and
\begin{align}\label{f1}
\int_{\varsigma}^{\infty}\frac{\log\Delta(t)}{t^{1+\frac{1}{\sigma}}}dt<\infty,\ \sigma>1
\end{align}
with $\varsigma$ being positive constant close to 0.
The function $\Delta$ has an inverse because it is strictly increasing. If $\omega$ is non-resonant and satisfies \eqref{i1} and \eqref{f1}, it is said to satisfy the
$\sigma$-Bruno-R\"{u}ssmann condition. The frequency set $\Omega_{\sigma}$ consists of frequencies $\omega$ that meet $\sigma$-Bruno-R\"{u}ssmann condition.

Let $\Omega_{\tau}$ be the set of $\tau$- Diophantine vectors $(\tau>d-1)$, typically, $\Delta(|k|)\leq |k|^{\tau}$ for all $0\neq k\in\mathbb{Z}^{d}$. $\Omega_{\tau}$ is obviously non-empty and has full measure, if $\tau>d-1$. As defined, we have $\Omega_{\tau}\subset\Omega_{\sigma}$ (see \cite{MR4050197}). Thus, $\Omega_{\sigma}$ is non-empty.

Alternatively, from the perspective of the regularity of the Hamiltonian $H$, in the early 2000s, Popov \cite{MR2104602} challenged the traditional view on Hamiltonian regularity, relaxing the analytic requirement to Gevrey regularity (as defined in Definition \ref{at}) and proving a version of the KAM theorem under this weaker regularity condition. This regularity also applies in canonical perturbation theory. Namely, we consider a one-parameter family of perturbations
\begin{align}\label{a13}
H(\theta,I;t)\in \mathcal{G}^{\rho,\rho,1}_{L_{1},L_{2},L_{2}}(\mathbb{T}^{d}\times D_{0}\times(-1,1)),\quad  H(\theta,I;0)=h_{0}(I).
\end{align}
For perturbed Hamiltonian functions satisfying Gverey regularity, under $\sigma$-Bruno-R\"{u}ssmann conditions, we can still construct a family of canonical transformations that simplify the perturbed Hamiltonian functions into a more tractable form.

\subsection{Main result}
Suppose that $M$ is a compact metric space, endowed with a semiclassical measure $\mu$ (see Definition \ref{a12}). The pseudodifferential operator
$P_{h}(t)$ satisfies the assumptions made earlier. Under the assumption of the existence of quantum Birkhoff normal form for $P_{h}(t)$ with $\sigma$-Bruno-R\"{u}ssmann condition, our main result is :
\begin{theorem}\label{ai}
Suppose that $M$ is a compact boundaryless $\mathcal{G}^{\rho}$ surface and $\omega\in\Omega_{\sigma}$. Let $P_{h}(t)$ is a family of self-adjoint elliptic semiclassical pseudo-differential operators acting on $\mathcal{C}^{\infty}(M;\Omega^{1/2})$ ( that is, $u\in C^{\infty}(M,\Omega^{\frac{1}{2}})$ if and only if $u=u(x)|dx|^{\frac{1}{2}}$) with fixed positive differential order $J$, such that
\begin{enumerate}[(i)]
  \item  $P_{h}(t)$ has full symbol real-valued and in the  $Gevrey$ class $\mathcal{S}_{l}(T^{*}M)$, where $l=(\sigma,\mu,\lambda,\bar{\rho}),$ with $\mu,\lambda\geq\rho+1,\ \rho\geq\sigma>1,\ \bar{\rho}\geq\mu+\sigma+\lambda-1$;
  \item the principal symbol of $P_{h}(t)$ is given by $p(x,\xi;0)\in \mathcal{G}^{\rho,\rho}(T^{*}M)$, and the subprincipal symbol is identically zero;
  \item the Hamiltonian $p_{0}(x,\xi)=p(x,\xi;0)$ is, in some open set of phase space $T^{*}M$, nondegenerate and completely integrable;
  \item written in action-angle coordinates $(\theta,I)\in \mathbb{T}^{d}\times D$ for the Hamiltonian $p_{0}$, the vector fields
\begin{align}\label{a4}
\nabla_{I}h_{0}(I), \ \nabla_{I}\left(\int_{\mathbb{T}^{d}}\partial_{t} H(\theta,I;0)d\theta \right), \cdots, \nabla_{I}\left(\int_{\mathbb{T}^{d}}\partial_{t}^{(d-1)} H(\theta,I;0)d\theta \right)\\ \notag
 \text {are linearly  independent for $I \in D$ and all $h<h_{0}$},
\end{align}
where $H(\theta,I;t)$ denotes $p(x,\xi;t)$ written in the action-angle coordinates for $p_{0}$, and $h_{0}(I):=H(\theta,I;0)$.
\end{enumerate}
Then there exists $t_{0}>0$ such that for almost all $t\in [0,t_{0}]$, and for almost all $KAM$ tori $\Lambda_{\omega}=\mathbb{T}^{d}\times \{I_{\omega}\}$ with $\omega\in\Omega_{\sigma}$, we can find a semiclassical measure $\nu$ associated to the eigenfunctions of $P_{h}(t)$ that has positive mass on $\Lambda_{\omega}$.
\end{theorem}

In our work, under the assumptions of Theorem \ref{ai}, we identify $d$ linearly independent quantities, as shown in Proposition \ref{am}. This is because the $d$ quantities appearing in Theorem \ref{ai} are the leading order terms for the quasi-eigenvalues and their derivatives with respect to $t$ at $t=0$, respectively. This gives us the quasi-eigenvalue as a function of $t$, a key result discussed in Section 3. Building on the quantum Birkhoff normal form from Theorem \ref{ac}, we investigate the dependence of the integrable term $K_{0}(I;t)$ on $t$. In classical perturbation theory, a series of transformations $\chi$ with generating functions $\Phi$ can be found such that $H\circ\chi=K_{0}+R_{0}$, causing the integrable part $K_{0}$ to evolve as a power series in $t$. The final result is presented in Equation \eqref{r1}. The determination of the generating function $\Phi$ requires solving the homological equation for the frequency vector $\omega\in\Omega_{\sigma}$ , under a Gevrey regularity condition. To the best of my knowledge, this approach is relatively novel. This overcomes the difficulties related to the dimension. The dimension is exactly equal to $d$, the maximal dimension of action variables in the system \eqref{a13}.

In conclusion, we not only extend the results on two-dimensional tori from \cite{MR4404789} to the high-dimensional tori, but also relax the frequency condition to our newly proposed $\sigma$-Bruno-R\"{u}ssmann. Here, we gratefully acknowledge the important contributions made by Gomes.

\subsection{Structure of the work}
In section 2, we recall that the construction of the QBNF of $P_{h}(t)$. Furthermore, we provide that how to construct quasimodes using QBNF (quantum Birkhoff normal form). In section 3, we mainly focus on the expression of $K_{0}(I;t)$ (the integrable part of BNF of Hamiltonian $H$) with respect to the variable $t$ in the spirit of Gallavotti \cite{MR733479} under $\sigma$-Bruno-R\"{u}ssmann condition. The final result is presented in Lemma \ref{t1}. Next, in section 4, we will prove Theorem \ref{ai}. Firstly, we obtain the separability of two distinct quasieigenvalues, that is, for any two distinct $\mu_{m},\mu_{m^{\prime}}$, their distance must be at least a multiple of $h^{3/2}$ by using the non-resonant condition. When $\gamma>7/4+2d$, we find that for $\forall m\in \mathcal{M}_{h}(t), m^{\prime}\in \mathbb{Z}^{d}$, with $m\neq m^{\prime}$, we have $dist\ (\mu_{m},\mu_{m^{\prime}})>Ch^{\gamma}$. It is worth noting that $\gamma$ depends on the dimension $d$. In other words, we construct a large number of the energy windows $\mathcal{W}_{m}(h), m\in\mathcal{M}_{h}(t)$ (appearing in the sequel) such that $\mu_{m}$ is the only quasieigenvalue in $\mathcal{W}_{m}(h)$. In section 4.2, we study the number of eigenvalues in all the energy windows, which is finite. Applying spectral theory, we present that the lower bound the maximal of $|\langle u,v\rangle|$ is a constant independent of $h$, as shown in Lemma \ref{ar}, where $v$ is quasimode associated with the quasieigenvalue $\mu$, and $u$ is the eigenfunction associated to the eigenvalues lying in the energy window. So, we can extract a subfamily, indexed by a sequence $h_{j}$ tending to zero, associated to quasimodes that concentrate on almost every invariant KAM torus $\Lambda_{\omega}$. We can thus choose eigenfunctions $u(h_{j})$ so that $|\langle u(h_{j}),v(h_{j})\rangle|$ has a lower bound given by a positive constant. So, we obtain a sequence of eigenfunctions $u(h_{j})$ with positive mass on $\Lambda_{\omega}$.

\section{Preliminaries}

\subsection{Quasimodes}
For the semiclassical family of pseudodifferential operators $P_{h}(t)$, Popov constructed a quantum Birkhoff normal form (shortly written as QBNF) $P_{h}^{0}(t)$ from Birkhoff normal form of Hamiltonian $H$\ (see \cite{MR1770799},\cite{MR1770800}). Extensive use has been made of QBNF when estimating eigenvalues and eigenfunctions. It is worth mentioning that QBNF can not only be used to derive the definition of quasimodes, but also applied to inverse spectral problems , particularly in the study of wave trace invariants \cite{MR1390650, MR1909649}. From a QBNF, Popov derived a semiclassical family of quasimodes whose definition is provided below:
\begin{definition}\label{a9}
Given $\rho>1$, we define a $\mathcal{G}^{\rho}$ quasi-mode of $P_{h}(t)$ as
\begin{align*}
\mathcal{Q}=\left\{(u_{m}(x;t,h),\lambda_{m}(t;h)):m\in \mathcal{M}_{h}(t)\right\},
\end{align*}
where $u_{m}(x;t,h)\in C_{0}^{\infty}(M)$ has a support in a fixed bounded domain independent of $h,\lambda_{m}(t;h)$ are real valued functions of\ $0<h\leq h_{0},\ 0<t\leq t_{0},\ \mathcal{M}_{h}(t)$ is a finite index set for each fixed $h$, and
\begin{enumerate}[(i)]
  \item $\parallel P_{h}u_{m}-\lambda_{m}(t,h)u_{m}\parallel_{L^{2}}\leq C\exp(-c/h^{1/\rho}),\quad m\in\mathcal{M}_{h}(t)$,
  \item $|\langle u_{m},u_{l}\rangle_{L^{2}}-\delta_{m,l}|\leq C\exp(-c/h^{1/\rho}),\quad m,l\in\mathcal{M}_{h}(t)$.
\end{enumerate}
Here $C$ and $c$ are positive constants, and $\delta_{m,l}$ is the $Kronecker$ index.
\end{definition}

The Gevrey micro-support of such Gevrey Quasimodes coincides with the union $\Lambda$ of a suitable Gevrey family of KAM invariant tori $\Lambda_{\omega}$ (see \cite{MR2111816} for the definition of the Gevrey micro-support). This naturally raises a question which was posed by S. Zelditch about a decade ago: whether the true eigenfunctions behave similarly.

Recall that for any $C^{\infty}$ quasimode $\mathcal{Q}$ the right hand side in the above inequality is $O_{N}(h^{N})$ for each $N\geq 0$.
Lazutkin \cite{MR1239173} derived quasimodes for the Laplace operator in strictly convex bounded domains within $\mathbb{R}^{2}$, associated with a Cantor family of invariant tori. In our work, we require that $P_{h}(t)$ is self-adjoint operator. Similarly, the reference \cite{MR4487490} gave the construction of quasimodes for non-selfadjoint operators.

\subsection{Quantum Birkhoff normal form}
Let $M$ be a compact $\mathcal{G}^{\rho}$-smooth manifold of dimension $d\geq2$. The function $p(x,\xi;t)$, the principal symbol of $P_{h}(t)$, is the Hamiltonian function in the classical sense, then  from the Liouville-Arnold theorem \cite{MR0997295}, there exists a coordinate transformation denoted by $\chi_{1}: \mathbb{T}^{d}\times D\rightarrow T^{*}M$, such that $H(\theta,I;t)=(p\circ\chi_{1})(x,\xi;t)$.
\par
We consider the system \eqref{a13} with the perturbation norm
\begin{align*}
\varepsilon_{H}:=\kappa^{-2}\|H-h_{0}\|_{L_{1},L_{2},L_{2}},
\end{align*}
where $h_{0}(I)=H(\theta,I;0)$ and $\kappa$ is a positive constant. As mentioned in the introduction, $P_{h}(t)$ is a semiclassical operators family which is quantization of $H$ with full symbols in the Gevrey class $\mathcal{S}_{l}$, acting on half-densities with principal symbol $p$ and vanishing subprincipal symbol. Now fix  $l=(\sigma,\mu,\lambda,\bar{\rho})$ with $\mu,\lambda\geq\rho+1,\ \rho\geq\sigma>1,\ \bar{\rho}\geq\mu+\sigma+\lambda-1$. We have:
\begin{theorem}\label{ac}
Suppose that there exists a real analytic exact symplectic map $\chi_{1}:\mathbb{T}^{d}\times D\rightarrow U\subset T^{*}M $ so that the Hamiltonian $\hat{H}=H(\chi_{1}(\theta,I)),\ (\theta,I)\in \mathbb{T}^{d}\times D$, which can be written in the following form
\begin{align*}
\hat{H}(\theta,I;t)=K(I;t)+R(\theta,I;t).
\end{align*}
Furthermore, suppose that the frequency $\omega$ satisfies $\sigma$-Bruno-R\"{u}ssmann condition.
Then there exists a family of uniformly bounded $h-FIOs$ (defined in \eqref{bu}) $U_{h}(t): L^{2}(\mathbb{T}^{d};\mathbb{L})\rightarrow L^{2}(M), 0<h\leq h_{0}$, associated with the canonical relation $C=graph(\chi),\chi=\chi_{0}\circ\chi_{1}$ such that the following statements hold:
\begin{enumerate}[(i)]
  \item  $U^{*}_{h}(t)U_{h}(t)-Id$ is a pseudo-differential operator with a symbol in the $Gevrey$ class $\mathcal{S}_{l}(\mathbb{T}^{d}\times D)$ which is equivalent to 0 on the $\mathbb{T}^{d}\times D_{0}$, where $D_{0}$ is a sub-domain of $D$ contains $E_{\kappa}(t)=\omega^{-1}(\Omega_{\sigma};t)$;
  \item $P_{h}(t)\circ U_{h}(t)=U_{h}(t)\circ P_{h}^{0}(t)$, and the full symbol $p^{0}(\varphi,I;t,h)$ of $P_{h}^{0}(t)$ has the form $p^{0}(\varphi,I;t,h)=K^{0}(I;t,h)+R^{0}(\varphi,I;t,h)$, where the symbols
  \begin{align*}
  K^{0}(I;t,h)=\sum_{0\leq j\leq\eta h^{-1/\rho}}K_{j}(I;t)h^{j}\ and \ R^{0}(\varphi,I;t,h)=\sum_{0\leq j\leq\eta h^{-1/\rho}}R_{j}(\varphi,I;t)h^{j}
  \end{align*}
  belong to the Gevrey class $\mathcal{S}_{l}(T^{*}\mathbb{T}^{d}),\ \eta>0$ is a constant, $K^{0}$ is real valued, and $R^{0}$ is equal to zero to infinite order on the Cantor set $\mathbb{T}^{d}\times E_{\kappa}(t)\times(-1,1)$.
\end{enumerate}
\end{theorem}

\begin{remark}
Here, $U_{h}(t)$ is a unitary operator; we therefore have $P_{h}(t)\circ U_{h}(t)=U_{h}(t)\circ P_{h}^{0}(t)$. Otherwise, the statement is modified to the following form:
\begin{align*}
P_{h}(t)\circ U_{h}(t)=U_{h}(t)\circ P_{h}^{0}(t)+R_{h}(t),
\end{align*}
where the symbol of $R_{h}(t)$ belongs to $\mathcal{S}_{l}^{-\infty}(T^{*}\mathbb{T}^{d})$ (see the Definition \eqref{bw}).
\end{remark}
The next conclusion is established on the basis of Theorem \ref{ac}.  Firstly, we suppose that Hamiltonian $H$ is analytic on time  parameter $t$.  We can construct a $\mathcal{G}^{\rho}-$quasimode $\mathcal{Q}$ of $P_{h}(t)$ with an index set
\begin{align}\label{a5}
\mathcal{M}_{h}(t)=\{m\in\mathbb{Z}^{d}:|E_{\kappa}(t)-h(m+\vartheta/4)|\leq Lh\}
\end{align}
for a fixed constant $L>0$, where $E_{\kappa}(t)=\omega^{-1}(\Omega_{\sigma};t)$ and $\vartheta$ is the Maslov class of the embedded Lagrangian tori in $T^{*}M$.  It is easy to check that
\begin{align*}
\sharp\ \{m\in \mathcal{M}_{h}(t)\}&=\frac{1}{(2\pi h)^{d}}Vol(\mathbb{T}^{d}\times D)(1+o(1))\\
& =\frac{1}{(2\pi h)^{d}}Vol(\Lambda)(1+o(1)),
\end{align*}
where $Vol(\Lambda)$ represents the Lebesgue measure of the union $\Lambda$ of the invariant tori in $T^{*}M$.

Moreover,  QBNF is used to define Quasimodes with exponentially small error term as shown in the following.
\begin{corollary}\label{br}
Suppose that $\omega\in \Omega_{\sigma}$.
For $m\in \mathcal{M}_{h}(t)$, let $u_{m}(x;t,h)=U_{h}(t)(e_{m})(x)$ and $\lambda_{m}(t;h)=K^{0}(h(m+\vartheta/4);t,h),\ t\in(0,t_{0})$. Then
\begin{align}\label{bl}
\mathcal{Q}=\{(u_{m}(x;t,h),\lambda_{m}(t;h)):m\in \mathcal{M}_{h}(t)\}
\end{align}
is a $\mathcal{G}^{\rho}$-quasimode of $P_{h}(t)$ in view of quantum Birkhoff normal form appearing in Theorem \ref{ac}, where $\{e_{m}\}_{m\in \mathbb{Z}^{d}}$ is a set of orthogonal bases for a Hilbert space $L^{2}(\mathbb{T}^{d};\mathbb{L})$ associated with quasiperiodic functions
\begin{align*}
\tilde{e}_{m}(x):=\exp(i(m+\vartheta/4)\cdot x)
\end{align*}
on $\mathbb{R}^{d}$, where $i=\sqrt{-1}$.
\end{corollary}
\begin{proof}
We write $P_{h}^{0}(t)=K_{h}^{0}(t)+R_{h}^{0}(t)$, where the symbols of $K_{h}^{0}(t)$ and $R_{h}^{0}(t)$ are respectively $K^{0}(I;t,h)$ and $R^{0}(\varphi,I;t,h)$ . From the definition of the functions $e_{m}$, it is easy to see that
\begin{align}\label{x1}
P_{h}^{0}(t)(e_{m})(\varphi) & =(K_{h}^{0}(t)+R_{h}^{0}(t))(e_{m})(\varphi)\\ \notag
& = (K^{0}(h(m+\vartheta/4);t,h)+R^{0}(h(m+\vartheta/4),\varphi;t,h))e_{m}(\varphi)\\ \notag
& =(\lambda(t,h)+R^{0}(\varphi,h(m+\vartheta/4);t,h))e_{m}(\varphi)
\end{align}
for any $m\in\mathcal{M}_{h}(t)$. On the other hand, we expand $\partial_{I}^{\alpha}\partial_{\varphi}^{\beta}R^{0}(\varphi,I_{m};t,h)$ in Taylor series at some $I_{0}\in E_{\kappa}(t)$ such that $|I_{0}-I_{m}|=|E_{\kappa}(t)-I_{m}|=\inf_{I^{\prime}\in E_{\kappa}(t)}|I^{\prime}-I_{m}|$,  where $|E_{\kappa}(t)-I_{m}|=\inf_{I^{\prime}\in E_{\kappa}(t)}|I^{\prime}-I_{m}|$ is the distance to the compact set $E_{\kappa}(t)$. We get for any $\alpha,\beta\in\mathbb{Z}_{+}^{d}$ and $n\in \mathbb{Z}_{+}$,
\begin{align}\label{a11}
|\partial_{I}^{\alpha}\partial_{\varphi}^{\beta}R^{0}(\varphi,I;t,h)|\leq C^{|\alpha|+|\beta|+n+1}\alpha!^{\rho}\beta!^{\rho+1}n!^{\rho}|E_{\kappa}(t)-I_{m}|^{n}
\end{align}
for all $(\varphi,I,t,h)\in  \mathbb{T}^{d}\times D,\ I_{m} \notin E_{\kappa}(t),\ h\in(0,h_{0}],\ t\in(0,t_{0}]$. Using Stirling's formula, we may minimize the  right-hand side with respect to $n\in \mathbb{Z}_{+}$. An optimal choice for $n$ will be
\begin{align}\label{a10}
n\sim (-C^{-1}|E_{\kappa}(t)-I_{m}|)^{-\frac{1}{\rho}},
\end{align}
such that
\begin{align*}
|\partial_{I}^{\alpha}\partial_{\varphi}^{\beta}R^{0}(\varphi,I;t,h)|\leq C^{|\alpha|+|\beta|}\alpha!^{\rho}\beta!^{\rho+1}\exp(-C^{-1}|E_{\kappa}(t)-I_{m}|^{-\frac{1}{\rho}})
\end{align*}
for every $\alpha,\ \beta\in\mathbb{Z}_{+}^{d}$ and $(\varphi,I;t,h)\in  \mathbb{T}^{d}\times D\times(0,t_{0}]\times(0,h_{0}],I_{m}\notin E_{\kappa}(t)$.
From definition \eqref{a5} of the index set $\mathcal{M}_{h}(t)$, we can derive that
\begin{align}\label{y1}
|\partial_{I}^{\alpha}\partial_{\varphi}^{\beta}R^{0}(\varphi,I;t,h)|\leq C^{|\alpha|+|\beta|+1}\alpha!^{\rho}\beta!^{\rho+1}\exp(-ch^{-\frac{1}{\rho}}),
\end{align}
where the constant $c$ depends on $C^{-1}, L,\rho$.
Here, we explain why the above form \eqref{a10} is the optimal choice for $n$. The terms on the right side of inequality \eqref{a11} that involves $n$ are as follows:
\begin{align*}
f(n)\triangleq C^{n}n!^{\rho}|E_{\kappa}(t)-I_{m}|^{n}.
\end{align*}
Taking the logarithm of both sides, we have
\begin{align*}
\log f(n)= n\log C+\rho\log n!+n\log|E_{\kappa}(t)-I_{m}|.
\end{align*}
Using the logarithmic form of Stirling's approximation, that is,
\begin{align*}
\log n!\sim n\log n-n,
\end{align*}
one obtains
\begin{align*}
\log f(n)= n(\log C+\log|E_{\kappa}(t)-I_{m}|-\rho)+\rho n\log n.
\end{align*}
To minimize $\log f(n)$, take the derivative with respect to $n$ and set the derivative to zero.
By using Proposition \ref{as}, we prove that $\mathcal{Q}$ satisfies $(\rm{i})$ in definition of the quasi-mode from \eqref{x1} and \eqref{y1}. By using the fact that $U_{h}(t)$ is unitary from Theorem \ref{ac}, and that $\{e_{m}\}_{m\in \mathbb{Z}^{d}}$ are exactly orthogonal by construction, one shows that $U_{h}(t)(e_{m})$ is almost-orthogonal, consequently, this proves that $\mathcal{Q}$ satisfies $(\rm{ii})$ in definition of the quasi-mode.
\end{proof}
The proof follows a similar approach to that of Popov \cite{MR1770800}.
\begin{remark}
The quasimode as stated above with remainder of order $O(\exp(-c/h^{1/\rho}))$ is far sharp than we need in this paper. We required only the truncated version, with error terms of order $O(h^{\gamma+1})$ for some fixed $\gamma>0$. That is, we can find a family of uniformly bounded $h-FIOs$ (defined in \eqref{bu}) $U_{h}(t)$ such that
\begin{align*}
P_{h}(t)\circ U_{h}(t)=U_{h}(t)\circ P_{h}^{0}(t)+R_{h}(t)\in h^{\gamma+1}\mathcal{S}_{l}.
\end{align*}
\end{remark}
\begin{remark}
By truncating the expansion of the elliptic symbol $a$ in Theorem \ref{ac} to the finite order error $O(h^{\gamma+1})$, we find that the QBNF symbols $K^{0}$ and $R^{0}$ also admit expansions  to the same finite order without any loss. But this could weaken the error estimate in the quasimode \eqref{bl} to $O(h^{\gamma+1})$. It is sufficient to use quasimodes with error order $O(h^{\gamma+1})$ in this section.
\end{remark}

\section{The expression of $K_{0}(I;t)$}
In this section, we prepare for the proof of Proposition \ref{am} in the next section, specifically deriving the expression of $K_{0}$ in terms of the variable $t$. The basic idea is that we define a series of canonical transformations \(\chi\) parameterized by the time variable $t$, such that $H \circ \chi = K_{0} + R_{0}$, which is the Birkhoff normal form of the perturbed Hamiltonian $H$. And the powers of $K_{0}(I;t)$ in terms of $t$ are successively accumulated.  The integrable part $K_{0}$ of this normal form is the expression we seek.

\subsection{The classical perturbation theory under \texorpdfstring{$\sigma$}--Bruno-R\"{u}ssmann condition}
In this section, for any given Hamiltonian $H$, we perform a change of variables that transforms the initial problem into an equivalent one involving a new function
$H^{1}=H\circ\chi$. The transformation $\chi$ is chosen such that, after a few iterations, the Hamiltonian
\begin{align*}
H^{r}=H\circ \chi^{r}
\end{align*}
becomes more accessible to treatment in the spirit of Gallavotti,G. \cite{MR733479}. In \cite{MR0872139}, the authors proved that the algorithm of classical perturbation theory for the Schr\"{o}dinger operator.
\subsubsection{Classical theory based on the Hamilton-Jacobi Method}
We consider a perturbed Hamiltonian $H$ as follows:
\begin{align*}
H(\varphi,I;t)=h_{0}(I)+f_{0}(\varphi,I;t),\ t\in(-1,1)
\end{align*}
with
\begin{align*}
h_{0}(I)=H(\varphi,I;0)
\end{align*}
and
\begin{align*}
f_{0}(\varphi,I;t)=H(\varphi,I;t)-h_{0}(I)
&=t\cdot\partial_{t}H(\varphi,I;0)+\int_{0}^{t}(1-s)\partial_{t}^{2}H(\varphi,I;s)ds\\
&=t\cdot\partial_{t}H(\varphi,I;0)+O(t^{2}).
\end{align*}
And we assume that $f_{0}\in \mathcal{G}^{\rho,\rho,1}$ and $\|f_{0}\|_{L_{1},L_{2},L_{2}}=O(t)$ with $L_{2}>L_{1}\geq1$. The corresponding Gevrey norm is defined as shown in Equation \eqref{l1}. We define a canonical transformation $\chi$ parameterized by $t$ with a generating function
\begin{align*}
\Phi(\varphi,I;t)=\langle I,\varphi\rangle+t\cdot\psi(\varphi,I)
\end{align*}
by
\begin{equation}\label{a8}
\begin{split}
\left\{
\begin{aligned}
I &= I^{\prime} + t \cdot \frac{\partial \psi}{\partial \varphi}(I^{\prime}, \varphi), \\
\varphi^{\prime} &= \varphi + t \cdot \frac{\partial \psi}{\partial I^{\prime}}(I^{\prime}, \varphi).
\end{aligned}
\right.
\end{split}
\end{equation}
Therefore, we can immediately conclude that $(I,\varphi)=\chi(I^{\prime},\varphi^{\prime})$. Combining with \eqref{a8}, one can find an equation for $\psi$:
\begin{align}\label{c1}
h_{0}\left(I^{\prime}+t\cdot\frac{\partial\psi}{\partial\varphi}(I^{\prime},\varphi)\right)&+ f_{0}\left(I^{\prime}+t\cdot\frac{\partial\psi}{\partial\varphi}
(I^{\prime},\varphi),\varphi;t\right)\\ \notag
&=\{\varphi-independent\ and\ higher\ order\ w.r.t.\ t\}.
\end{align}
Here, the higher order is  at least second order with respect to $t$.
Notice that the left hand side of \eqref{c1} is already of $O(t)$, and we expand \eqref{c1} neglecting $O(t^{2})$, it is easily to find that
\begin{align}\label{c2}
\omega_{0}(I^{\prime})\cdot t\frac{\partial\psi}{\partial\varphi}(I^{\prime},\varphi)+f_{0}(I^{\prime},\varphi;t)
=\{\varphi-independent\}.
\end{align}
The right hand side of \eqref{c2} is the average of $f_{0}$ over the angles variable $\varphi$ denoted by $f_{00}(I^{\prime};t)$, and the above equation about the function $\psi$ simplifies to
\begin{align}\label{d1}
\omega_{0}(I^{\prime})\cdot t\frac{\partial\psi}{\partial\varphi}(I^{\prime},\varphi)+f_{0}(I^{\prime},\varphi;t)- f_{00}(I^{\prime};t)=0.
\end{align}
The function $\psi$ should be determined by this equation for $\omega_{0}\in \Omega_{\sigma}$. Furthermore, we can define the canonical transformation $\chi$ such that
\begin{equation}
H_{0}(\chi(I^{\prime},\varphi^{\prime};t))=h_{0}(I^{\prime})+ f_{00}(I^{\prime};t)+  \text{the higher-order terms with respect to } t.
\end{equation}
The present challenge, of course, lies in determining whether the first-order Hamilton-Jacobi equation \eqref{d1} admits a solution.
\par
In the sequel, we will focus on solving Equation \eqref{d1} and providing an estimate for the function $\psi$. To do this, we firstly address the equation $\mathcal{L}_{\omega}\psi=f$ under the $\sigma$-Bruno-R\"{u}ssmann condition in Gevrey classes within $E_{\kappa}$. For Equation \eqref{d1}, in order to eliminate the variable $t$, let $f=t^{-1}(f_{00}-f_{0})$. Since $f_{0}$ and  $f_{00}$ are linear in $t$, it follows that $f$ is independent of $t$. Therefore, we obtain
\begin{align}\label{m1}
\left|D_{I}^{\alpha}D_{\varphi}^{\beta}f(\varphi,I)\right|
\leq & t C_{1}^{|\alpha|}C_{2}^{|\beta|}\alpha!^{\mu}\beta!^{\sigma}+t C_{1}^{|\alpha|}\alpha!^{\mu}\\ \notag
\leq & 2t C_{1}^{|\alpha|}C_{2}^{|\beta|}\alpha!^{\mu}\beta!^{\sigma}
\end{align}
with $C_{1}=CL_{1},C_{2}=CL_{2}$ and $C_{1},C_{2}\geq1, C>0$.
We are going to solve the equation
\begin{align}\label{h1}
 \mathcal{L}_{\omega}\psi(\varphi,I)=f(\varphi,I)
\end{align}
with
\begin{align}\label{a6}
 \int_{\mathbf{T}^{d}}f(\varphi,I){\rm d}\varphi=0,
\end{align}
and provide the upper bound of $\left|D_{I}^{\alpha}D_{\varphi}^{\beta}\psi(\varphi,I)\right|$ for any $\alpha,\beta\in \mathbb{Z}_{+}^{d}$.
\par
\begin{proposition}\label{j1}
Let $f\in C^{\infty}(\mathbb{T}^{d}\times E_{\kappa})$ satisfy \eqref{m1},\eqref{a6}. Then  equation \eqref{h1} has a unique solution $\psi \in  C^{\infty}(\mathbb{T}^{d}\times E_{\kappa})$ and the solution $\psi$ of equation \eqref{h1} satisfies the estimate
\begin{align}\label{k1}
\left|D_{I}^{\alpha}D_{\varphi}^{\beta}\psi(\varphi,I)\right|\leq t DC_{1}^{|\alpha|}C_{2}^{|\beta|}\alpha!^{\mu}\beta!^{\rho}
\end{align}
for $I\in E_{\kappa}$ and $\alpha,\alpha_{1},\beta\in \mathbb{Z}_{+}^{d},\ \alpha_{1}\leq\beta$.
\end{proposition}
For the proof of this proposition can be found in the appendix. We have proved the solution $\psi$ of the equation \eqref{d1} exists and is $O(t)$.
At present, we find the generating function $\psi_{0}$, so that
\begin{align*}
\Phi_{0}(\varphi,I)=\langle I,\varphi\rangle+t\cdot\psi_{0}(\varphi,I),
\end{align*}
and the corresponding canonical transformation is $\chi_{0}: D_{1}\times \mathbb{T}^{d}\rightarrow D\times \mathbb{T}^{d}$. The definition of $\psi_{0}$, motivated by perturbation theory, is
\begin{align*}
\psi_{0}(I^{\prime},\varphi)=\sum_{\gamma\in Z^{d}\backslash\{0\}}\frac{f_{\gamma}(I^{\prime})}{-i\langle\omega_{0}(I^{\prime}),\gamma\rangle}e^{i\langle\gamma,\varphi\rangle},
\end{align*}
where $i=\sqrt{-1}$.
\par
Next we define a map :\ $(I^{\prime},\varphi)\rightarrow(I,\varphi^{\prime})$ as
\begin{equation*}
\begin{split}
 \left \{
\begin{aligned}
I&=I^{\prime}+t\cdot\frac{\partial\psi_{0}}{\partial\varphi}(I^{\prime},\varphi),\\
\varphi^{\prime}&=\varphi+t\cdot\frac{\partial\psi_{0}}{\partial I^{\prime}}(I^{\prime},\varphi).
\end{aligned}
\right.
\end{split}
\end{equation*}
We try to look for inversions of the respective form
\begin{equation}
\begin{aligned}
\begin{split}
 \left \{
\begin{aligned}
I^{\prime}&=I+\Xi^{\prime}(I,\varphi;t),\\
\varphi&=\varphi^{\prime}+\Lambda(I^{\prime},\varphi^{\prime};t).
\end{aligned}
\right.
\end{split}
\end{aligned}
\end{equation}
To achieve this goal, we have to use a variant of the Komatsu implicit function theorem, namely, Proposition \ref{s1}. Furthermore, the  following identities should be satisfied:
\begin{equation}
\begin{aligned}
\begin{split}
 \left \{
\begin{aligned}
\Xi^{\prime}(I,\varphi;t)&=-t\cdot\frac{\partial^{|\mathbf{v}|}\psi_{0}}{\partial\varphi}(I^{\prime},\varphi)=-t\cdot\frac{\partial\psi_{0}}{\partial\varphi}(I^{\prime},
\varphi),\\
\Lambda(I^{\prime},\varphi^{\prime};t)&=-t\cdot\frac{\partial^{|\mathbf{v}|}\psi_{0}}{\partial I^{\prime}}(I^{\prime},\varphi)=-t\cdot\frac{\partial\psi_{0}}{\partial I^{\prime}}(I^{\prime},\varphi).
\end{aligned}
\right.
\end{split}
\end{aligned}
\end{equation}
We can obtain the Gevrey estimate of the functions $\Xi,\Xi^{\prime},\Lambda$ as stated below:
\begin{align*}
|D_{I^{\prime}}^{\alpha}D_{\varphi^{\prime}}^{\beta}\Xi(I^{\prime},\varphi^{\prime};t)|,\
|D_{I^{\prime}}^{\alpha}D_{\varphi^{\prime}}^{\beta}\Xi^{\prime}(I^{\prime},\varphi^{\prime};t)| &\leq  tDC_{1}^{|\alpha|}C_{2}^{|\beta|+1}\alpha!^{\mu}(\beta+\mathbf{v})!^{\rho} \\
|D_{I^{\prime}}^{\alpha}D_{\varphi^{\prime}}^{\beta}\Lambda(I^{\prime},\varphi^{\prime};t)| &\leq t DC_{1}^{|\alpha|+1}C_{2}^{|\beta|}(\alpha+\mathbf{v})!^{\mu}\beta!^{\rho}
\end{align*}
for $\alpha,\beta\in \mathbb{Z}_{+}^{d}$. Notice that $\bf{v}$ is a vector, and we can check $|\mathbf{v}|=1$.
Then, it is possible that we can define a canonical transformation $\chi_{0}$ parameterized by $t$ as follows:
\begin{align*}
\chi_{0}: (I^{\prime},\varphi^{\prime})\rightarrow (I,\varphi)
\end{align*}
by
\begin{equation*}
\begin{split}
 \left \{
\begin{aligned}
I&=I^{\prime}+\Xi(I^{\prime},\varphi;t),\\
\varphi&=\varphi^{\prime}+\Lambda(I^{\prime},\varphi;t).
\end{aligned}
\right.
\end{split}
\end{equation*}

\subsubsection{The estimate of the remainder }
Under the above canonical transformation, we have
\begin{align*}
H^{1}(I^{\prime},\varphi^{\prime};t)&=\left(H_{0}\circ \chi_{0})(I^{\prime},\varphi^{\prime};t\right)\\
&=h_{0}\left(I^{\prime}+\Xi(I^{\prime},\varphi^{\prime};t)\right)+ f_{0}\left(I^{\prime}+\Xi(I^{\prime},\varphi^{\prime};t),\varphi^{\prime}+\Lambda (I^{\prime},\varphi^{\prime};t);t\right)\\
&\triangleq h_{1}+f_{1}.
\end{align*}
So, through simple calculations, we derive that
\begin{align*}
h_{1}(I^{\prime};t)=h_{0}(I^{\prime})+f_{00}(I^{\prime};t).
\end{align*}
It follows that
\begin{align*}
f_{1}\left(I^{\prime},\varphi^{\prime};t\right)&=H^{1}(I^{\prime},\varphi^{\prime};t)-h_{1}(I^{\prime};t)\\
&= f_{0}\left(I^{\prime}+\Xi(I^{\prime},\varphi^{\prime};t),\varphi^{\prime}+\Lambda (I^{\prime},z^{\prime};t);t\right)+h_{0}\left(I^{\prime}+\Xi(I^{\prime},\varphi^{\prime};t)\right)-h_{0}\left(I^{\prime}\right)\\
&\quad - f_{00}\left(I^{\prime};t\right).
\end{align*}
Recall that \eqref{d1}, that is to say,
\begin{align*}
-\omega_{0}(I^{\prime})\cdot\Xi(I^{\prime},\varphi^{\prime};t)+f_{0}\left(I^{\prime},\varphi^{\prime}+\Lambda (I^{\prime},\varphi^{\prime};t);t\right)-f_{00}(I^{\prime};t)=0.
\end{align*}
Then we get
\begin{align*}
&f_{1}\left(I^{\prime},\varphi^{\prime};t\right)=H^{1}(I^{\prime},\varphi^{\prime};t)-h_{1}(I^{\prime};t)\\
&=f_{0}\left(I^{\prime}+\Xi(I^{\prime},\varphi^{\prime};t),\varphi^{\prime}+\Lambda(I^{\prime},\varphi^{\prime};t);t\right)+h_{0}\left(I^{\prime}
+\Xi(I^{\prime},\varphi^{\prime};t)\right)
-h_{0}\left(I^{\prime}\right)-f_{00}\left(I^{\prime};t\right)\\
&=f_{0}\left(I^{\prime}+\Xi(I^{\prime},\varphi^{\prime};t),\varphi^{\prime}+\Lambda(I^{\prime},\varphi^{\prime};t);t\right)- f_{0}\left(I^{\prime},\varphi^{\prime}+\Lambda (I^{\prime},\varphi^{\prime};t);t\right)\\
&\quad +h_{0}\left(I^{\prime}+\Xi(I^{\prime},\varphi^{\prime};t)\right)
-h_{0}\left(I^{\prime}\right)- f_{00}\left(I^{\prime};t\right)+f_{0}\left(I^{\prime},\varphi^{\prime}+\Lambda (I^{\prime},\varphi^{\prime};t);t\right)\\
&=h_{0}\left(I^{\prime}+\Xi(I^{\prime},\varphi^{\prime};t)\right)-h_{0}\left(I^{\prime}\right)-\omega_{0}(I^{\prime})\cdot\Xi(I^{\prime},\varphi^{\prime};t)
\\&\quad + f_{0}\left(I^{\prime}+\Xi(I^{\prime},\varphi^{\prime};t),\varphi^{\prime}+\Lambda(I^{\prime},\varphi^{\prime};t);t\right)- f_{0}\left(I^{\prime},\varphi^{\prime}+\Lambda (I^{\prime},\varphi^{\prime};t);t\right)\\
&\triangleq f_{1}^{I}+ f_{1}^{II},
\end{align*}
which we  write as
\[
\begin{aligned}
f_{1}^{I}&=h_{0}\left(I^{\prime}+\Xi(I^{\prime},\varphi^{\prime};t)\right)-h_{0}\left(I^{\prime}\right)-\omega_{0}(I^{\prime})\cdot\Xi(I^{\prime},\varphi^{\prime};t)\\
&=\int_{0}^{1}(1-s)ds\frac{\partial^{2}h_{0}}{\partial I^{\prime 2}}\left(I^{\prime}+s\Xi(I^{\prime},\varphi^{\prime};t)\right)\Xi^{2}\left(I^{\prime},\varphi^{\prime};t\right),
\\
f_{1}^{II}&=f_{0}\left(I^{\prime}+\Xi(I^{\prime},\varphi^{\prime};t),
\varphi^{\prime}+\Lambda(I^{\prime},\varphi^{\prime};t);t\right)-f_{0}\left(I^{\prime},\varphi^{\prime}+\Lambda(I^{\prime},\varphi^{\prime};t);t\right)\\
&=\int_{0}^{1}ds
\frac{\partial f_{0}}{\partial I^{\prime}}\left(I^{\prime}+s\Xi(I^{\prime},\varphi^{\prime};t), \varphi+\Lambda(I^{\prime},\varphi^{\prime};t);t\right)\Xi\left(I^{\prime},\varphi^{\prime};t\right)
\end{aligned}
\]
with $\Xi(I^{\prime},\varphi^{\prime};t)=t\cdot\frac{\partial\psi_{0}}{\partial\varphi}(I^{\prime},\varphi;t)=-\Xi^{\prime}(I,\varphi;t)$, $\Lambda(I^{\prime},\varphi^{\prime};t)=-t\cdot\frac{\partial\psi_{0}}{\partial I^{\prime}}(I^{\prime},\varphi;t)$.
\par
So, we can bound $f_{1}$. According to the formula \eqref{l1}, we first calculate $|D_{I}^{\alpha}D_{\varphi}^{\beta}f_{1}|$ in order to gain the Gevrey norm of $f_{1}$. As a matter of fact, it is already known that
\begin{align*}
f_{1}^{I}=\int_{0}^{1}(1-s)ds\frac{\partial^{2}h_{0}}{\partial I^{2}}(I^{\prime}+s\Xi(I^{\prime},\varphi^{\prime};t))\Xi^{2}(I^{\prime},\varphi^{\prime};t);
\end{align*}
hence, it follows that
\begin{align*}
&D_{I^{\prime}}^{\alpha}D_{\varphi^{\prime}}^{\beta}f_{1}^{I}(I^{\prime},\varphi^{\prime};t)\\
&=\sum_{0<\alpha_{1}\leq\alpha}\sum_{0<\beta_{1}\leq\beta}\binom {\alpha}{\alpha_{1}}\binom{\beta}{\beta_{1}} D_{I^{\prime}}^{\alpha_{1}}
D_{\varphi^{\prime}}^{\beta_{1}}\frac{\partial^{2}h_{0}}{\partial I^{2}}\left(I^{\prime}+s\Xi(I^{\prime},\varphi^{\prime};t)\right)D_{I^{\prime}}^{\alpha-\alpha_{1}}D_{\varphi^{\prime}}^{\beta-\beta_{1}}\Xi^{2}(I^{\prime},\varphi^{\prime};t)\\
&=\sum_{0<\alpha_{1}\leq\alpha}\sum_{0<\beta_{1}\leq\beta}\binom {\alpha}{\alpha_{1}}\binom{\beta}{\beta_{1}}\Bigg[\frac{\partial^{|\alpha_{1}|+|\beta_{1}|+2}h_{0}(I^{\prime}+s\Xi(I^{\prime},\varphi^{\prime};t))}{\partial I^{\prime(|\alpha_{1}|+2)}\partial\varphi
^{\prime|\beta_{1}|}}\left(1+s\frac{\partial \Xi(I^{\prime},\varphi^{\prime};t)}{\partial I^{\prime}}\right)\\
&\quad \times s\frac{\partial \Xi(I^{\prime},\varphi^{\prime};t)}{\partial\varphi^{\prime}}+\frac{\partial^{|\beta_{1}|+2}h_{0}\left(I^{\prime}+s\Xi(I^{\prime},\varphi^{\prime};t)\right)}{\partial I^{\prime 2}\partial\varphi
^{\prime|\beta_{1}|}}s \frac{\partial^{|\alpha_{1}|+1} \Xi(I^{\prime},\varphi^{\prime};t)}{\partial I^{\prime|\alpha_{1}|}\partial\varphi^{\prime}}\Bigg]\\
&\quad \times\left(2D_{I^{\prime}}^{\alpha-\alpha_{1}}\Xi(I^{\prime},\varphi^{\prime};t)\cdot D_{\varphi^{\prime}}^{\beta-\beta_{1}}\Xi(I^{\prime},\varphi^{\prime};t)+2\Xi(I^{\prime},\varphi^{\prime};t) D_{I^{\prime}}^{\alpha-\alpha_{1}}D_{\varphi^{\prime}}^{\beta-\beta_{1}}\Xi(I^{\prime},\varphi^{\prime};t)\right).
\end{align*}
Since $h_{0}$ belongs to $\mathcal{G}^{\rho}$, $\Xi,\Lambda\in \mathcal{G}^{\mu,\rho}$, we get
\begin{align*}
D_{I^{\prime}}^{\alpha}D_{\varphi^{\prime}}^{\beta}f_{1}^{I}(I^{\prime},\varphi^{\prime};t)= O(t^{3})+O(t^{4}).
\end{align*}
In the same way, it can be shown that
\begin{align*}
D_{I^{\prime}}^{\alpha}D_{\varphi^{\prime}}^{\beta}f_{1}^{II}(I^{\prime},\varphi^{\prime};t)= O(t^{2})+O(t^{3})+O(t^{4}).
\end{align*}
Thus, we can get the final estimate of $f_{1}$ with respect to the variable $t$:
\begin{align*}
D_{I^{\prime}}^{\alpha}D_{\varphi^{\prime}}^{\beta}f_{1}(I^{\prime},\varphi^{\prime};t)= O(t^{2})+O(t^{3})+O(t^{4}).
\end{align*}
Therefore, it follows that
\begin{align*}
f_{1}=O(t^{2})+O(t^{3})+O(t^{4}).
\end{align*}
It is worth noting that $h_{0}$ is independent of $t$, while $h_{1}$ is related to $t$. Therefore, the following iteration starts from $r = 2$. Let$^{\prime}$s first consider the case when $r = 2$. If we repeat the above process, we can conclude that
\begin{align}\label{n1}
h_{2}(I^{\prime},\varphi^{\prime};t)=h_{1}(I^{\prime};t)+f_{1,0}(\varphi^{\prime},I;t),
\end{align}
and
\begin{align}\label{n2}
f_{2}(I^{\prime},\varphi^{\prime};t)= O(t^{4})+\cdots+O(t^{10}).
\end{align}
\par
We can make an assertion as stated follows:
\begin{lemma}\label{t1}
Suppose that there is a canonical transformation $\chi_{r-1}$ such that
\begin{align*}
H^{r}(I^{\prime},\varphi^{\prime};t)=(H^{r-1}\circ\chi_{r-1})(I^{\prime},\varphi^{\prime};t)\triangleq h_{r}(I^{\prime};t)+f_{r}(I^{\prime},\varphi^{\prime};t),\ r\geq2,
\end{align*}
where
\begin{align}\label{a3}
h_{r}(I^{\prime};t)=h_{r-1}(I^{\prime};t)+f_{r-1,0}(I^{\prime};t),\ r\geq2,
\end{align}
and the term $f_{r-1,0}$ is described as \eqref{z}.
Then, we have
\begin{align*}
f_{r}(I^{\prime},\varphi^{\prime};t)\sim O(t^{2^{r}})+\cdots+O(t^{2^{r+1}+2^{r-1}+\cdots+2^{1}}),\ r\geq2.
\end{align*}
\end{lemma}
\begin{proof}
We have proved that the assertion holds, when $r=2$, which is the formulas \eqref{n1},\\
\eqref{n2}.
We will prove the assertion by making use of mathematical induction. We assume that the assertion holds for $r-1$. That is to say,
\begin{align*}
f_{r-1}(I^{\prime},\varphi^{\prime};t)\sim O(t^{2^{r-1}})+O(t^{2^{r-1}+1})+\cdots+O(t^{2^{r}+2^{r-2}+\cdots+2^{1}}).
\end{align*}
We will prove that it also holds for $r$, i.e.
\begin{align*}
f_{r}(I^{\prime},\varphi^{\prime};t)\sim O(t^{2^{r}})+O(t^{2^{r}+1})+\cdots+O(t^{2^{r+1}+2^{r-1}+\cdots+2^{1}}).
\end{align*}
Following the expression \eqref{a3}, we have
\begin{align*}
\begin{split}
h_{r}(I^{\prime};t)\sim \left \{
\begin{array}{lll}
O(c)+O(t)+\cdots+O(t^{2^{r}+2^{r-2}+\cdots+2^{1}}),\ r\geq3,\\
O(t^{2})+O(t^{3})+O(t^{4}),\ r=2.
\end{array}
\right.
\end{split}
\end{align*}
Notice that
\begin{align*}
H^{r}(I^{\prime},\varphi^{\prime};t)
&=(H^{r-1}\circ\chi_{r-1})(I^{\prime},\varphi^{\prime};t)\\
&= h_{r-1}(I^{\prime}+\Xi(I^{\prime},\varphi^{\prime};t);t)+f_{r-1}(I^{\prime}+\Xi(I^{\prime},\varphi^{\prime};t),\varphi+\Lambda(I^{\prime},\varphi^{\prime};t);t)\\
&\triangleq h_{r}(I^{\prime};t)+f_{r}(I^{\prime},\varphi^{\prime};t).
\end{align*}
Then we have
\begin{align}\label{o1}
f_{r}(I^{\prime},\varphi^{\prime};t)&=H^{r}(I^{\prime},\varphi^{\prime};t)-h_{r}(I^{\prime};t)\\ \notag
&=h_{r-1}(I^{\prime}+\Xi(I^{\prime},\varphi^{\prime};t);t)+f_{r-1}(I^{\prime}+\Xi(I^{\prime},\varphi^{\prime};t),\varphi+\Lambda(I^{\prime},\varphi^{\prime};t);t)\\ \notag
&\quad-h_{r-1}(I^{\prime};t)-f_{r-1,0}(I^{\prime};t),
\end{align}
where
\begin{align}\label{z}
f_{r-1,0}(I^{\prime};t)=(2\pi)^{-d}\int_{\mathbb{T}^{d}}f_{r-1}(I^{\prime},\varphi^{\prime};t)d\varphi^{\prime}.
\end{align}
For the domain $D^{r-1}$, we have:
\begin{align*}
\mathbb{T}^{d}\times D^{r-1}=\chi_{r}(\mathbb{T}^{d}\times D^{r}).
\end{align*}
In fact, we have
\begin{align*}
D^{r}\subset D^{r-1}\subset\cdots\subset D^{0}=\chi_{0}(D^{1}).
\end{align*}
Notice that $\chi_{r-1}$ is a canonical transformation with a generating function $\Phi_{r-1}(I^{\prime},\varphi;t)=\langle I^{\prime},\varphi\rangle+t^{2^{r-1}}\cdot\psi_{r-1}(I^{\prime},\varphi^{\prime})$, and
\begin{equation}
\begin{split}\label{a2}
 \left \{
\begin{array}{ll}
-\Xi_{r}(I^{\prime},\varphi^{\prime};t)=-t^{2^{r-1}}\cdot\frac{\partial\psi_{r-1}}{\partial\varphi}(I^{\prime},\varphi),\\
-\Lambda_{r}(I^{\prime},\varphi^{\prime};t)=t^{2^{r-1}}\cdot\frac{\partial\psi_{r-1}}{\partial I^{\prime}}(I^{\prime},\varphi).
\end{array}
\right.
\end{split}
\end{equation}
We can obtain the Gevrey estimates of the functions $\Xi_{r},\Xi_{r}^{\prime},\Lambda_{r}$ as stated below:
\begin{align}\label{v1}
|D_{I^{\prime}}^{\alpha}D_{\varphi^{\prime}}^{\beta}\Xi_{r}^{\prime}(I^{\prime},\varphi^{\prime};t)|,\
|D_{I^{\prime}}^{\alpha}D_{\varphi^{\prime}}^{\beta}\Xi_{r}^{\prime}(I^{\prime},\varphi^{\prime};t)| &\leq  t^{2^{r-1}} DC_{1r}^{|\alpha|}C_{2r}^{|\beta|+1}\alpha!^{\mu}(\beta+\mathbf{v})!^{\rho} \\
|D_{I^{\prime}}^{\alpha}D_{\varphi^{\prime}}^{\beta}\Lambda_{r}(I^{\prime},\varphi^{\prime};t)| &\leq t^{2^{r-1}} DC_{1r}^{|\alpha|+1}C_{2r}^{|\beta|}(\alpha+\mathbf{v})!^{\mu}\beta!^{\rho}
\end{align}
for $\alpha,\beta\in \mathbb{Z}_{+}^{d}$ and $C_{1r}=C^{r}L_{1},\ C_{2r}=C^{r}L_{2}$.
We also know that $\psi$ is the solution of the following equation
\begin{align}\label{p1}
\omega_{r-1}(I^{\prime};t)\cdot\Xi_{r}(I^{\prime},\varphi^{\prime};t)+ f_{r-1}(I^{\prime},\varphi^{\prime};t)- f_{r-1,0}(I^{\prime};t)=0,
\end{align}
where
\begin{align*}
\omega_{r-1}(I^{\prime};t)=\frac{\partial h_{r-1}(I^{\prime};t)}{\partial I^{\prime}}.
\end{align*}
It is easy to check that the first order Hamilton-Jacobi equation admits a solution $\psi_{r-1}$. Combining \eqref{o1} with \eqref{p1}, we derive
\begin{align*}
f_{r}(I^{\prime},\varphi^{\prime},t)
&=h_{r-1}(I^{\prime}+\Xi_{r}(I^{\prime},\varphi^{\prime};t),t)-h_{r-1}(I^{\prime},t)-\omega_{r-1}(I^{\prime},t)\cdot\Xi_{r}(I^{\prime},\varphi^{\prime},t)\\
&+f_{r-1}(I^{\prime}+\Xi_{r}(I^{\prime},\varphi^{\prime};t),\varphi^{\prime}+\Lambda_{r}(I^{\prime},\varphi^{\prime};t);t) -f_{r-1}(I^{\prime},\varphi^{\prime}+\Lambda_{r} (I^{\prime},\varphi^{\prime};t);t)\\
&=\int_{0}^{1}(1-s)ds\frac{\partial^{2}h_{r-1}}{\partial I^{\prime 2}}(I^{\prime}+s\Xi_{r}(I^{\prime},\varphi^{\prime},t))\cdot\Xi_{r}^{2}(I^{\prime},\varphi^{\prime},t)\\
& +\int_{0}^{1}ds\frac{\partial f_{r-1}}{\partial I^{\prime}}(I^{\prime}+s\Xi_{r}(I^{\prime},\varphi^{\prime},t),\varphi^{\prime}+\Lambda_{r}(I^{\prime},\varphi^{\prime};t),t)\cdot\Xi_{r}(I^{\prime},\varphi^{\prime},t)
\end{align*}
by making use of Taylor's expansions of $h_{0},f_{0}$ with respect to the action variable. Then, it can be shown that
\begin{align*}
&D_{I^{\prime}}^{\alpha}D_{\varphi^{\prime}}^{\beta}f_{r}(I^{\prime},\varphi^{\prime};t)\\
&=\sum\limits_{\substack{0<\alpha_{1}\leq\alpha \\ 0<\beta_{1}\leq\beta}}
\binom {\alpha}{\alpha_{1}}\binom{\beta}{\beta_{1}}D_{I^{\prime}}^{\alpha_{1}}
D_{\varphi^{\prime}}^{\beta_{1}}\frac{\partial^{2}h_{r-1}}{\partial I^{\prime 2}}(I^{\prime}+s\Xi_{r};t)D_{I^{\prime}}^{\alpha-\alpha_{1}}D_{\varphi^{\prime}}^{\beta-\beta_{1}}\Xi_{r}^{2}\\
&\quad+\sum\limits_{\substack{0<\alpha_{2}\leq\alpha \\ 0<\beta_{2}\leq\beta}}\binom {\alpha}{\alpha_{2}}\binom{\beta}{\beta_{2}}D_{I^{\prime}}^{\alpha_{2}}
D_{\varphi^{\prime}}^{\beta_{2}}\frac{\partial f_{r-1}}{\partial I^{\prime}}(I^{\prime}+s\Xi_{r},\varphi+\Lambda_{r};t)
D_{I^{\prime}}^{\alpha-\alpha_{2}}D_{\varphi^{\prime}}^{\beta-\beta_{2}}\Xi_{r}\\
&=\sum\limits_{\substack{0<\alpha_{1}\leq\alpha \\ 0<\beta_{1}\leq\beta}}\binom {\alpha}{\alpha_{1}}\binom{\beta}{\beta_{1}}\Bigg[\frac{\partial^{|\alpha_{1}|+|\beta_{1}|+2}h_{r-1}(I^{\prime}+s\Xi_{r})}{\partial I^{\prime|\alpha_{1}|+2}\partial\varphi
^{\prime|\beta_{1}|}}\left(1+s\frac{\partial \Xi_{r}}{\partial I^{\prime}}\right)s\frac{\partial \Xi_{r}}{\partial\varphi^{\prime}} +\frac{\partial^{|\beta_{1}|+2}h_{r-1}}{\partial I^{\prime|2|}\partial\varphi
^{\prime|\beta_{1}|}}\\
&\quad(I^{\prime}+s\Xi_{r};t)\cdot s\frac{\partial^{|\alpha_{1}|+1} \Xi_{r}}{\partial I^{\prime|\alpha_{1}|}\partial\varphi^{\prime}}\Bigg] \times \left(2D_{I^{\prime}}^{\alpha-\alpha_{1}}\Xi_{r} D_{\varphi^{\prime}}^{\beta-\beta_{1}}\Xi_{r}+2\Xi_{r} D_{I^{\prime}}^{\alpha-\alpha_{1}}D_{\varphi^{\prime}}^{\beta-\beta_{1}}\Xi_{r}\right)\\
&\quad+\sum\limits_{\substack{0<\alpha_{2}\leq\alpha \\ 0<\beta_{2}\leq\beta}}\binom {\alpha}{\alpha_{1}}\binom{\beta}{\beta_{1}}\Bigg[\frac{\partial^{|\alpha_{2}|+|\beta_{2}|+1}f_{r-1}^{11}(I^{\prime}+\Xi_{r},
\varphi+\Lambda_{r})}{\partial I^{\prime|\alpha_{2}|+1}\partial\varphi^{\prime|\beta_{2}|}}\left(1+s\frac{\partial \Xi_{r}}{\partial I^{\prime}}\right)\ s\frac{\partial \Xi_{r}}{\partial\varphi^{\prime}}\\
&\quad+\frac{\partial^{|\alpha_{2}|+|\beta_{2}|+1}f_{r-1}^{21}(I^{\prime}+\Xi_{r},\varphi+
\Lambda_{r};t)}{\partial I^{\prime|\alpha_{2}|+1}\partial\varphi^{\prime|\beta_{2}|}}\left(1+s\frac{\partial \Xi_{r}}{\partial I^{\prime}}\right)\left(1+\frac{\partial \Lambda_{r}}{\partial\varphi^{\prime}}\right)\\
&\quad+
\frac{\partial^{|\alpha_{2}|+|\beta_{2}|+1}f_{r-1}^{22}(I^{\prime}+s\Xi_{r},\varphi+
\Lambda_{r};t)}{\partial I^{\prime|\alpha_{2}|+1}\partial\varphi^{\prime|\beta_{2}|}}\frac{\partial \Lambda_{r}}{\partial I^{\prime}}\left(1+\frac{\partial \Lambda_{r}}{\partial\varphi^{\prime}}\right)\\
&\quad+\frac{\partial^{|\alpha_{2}|+|\beta_{2}|+1}f_{r-1}^{12}(I^{\prime}+t\Xi_{r},\varphi^{\prime}+
\Lambda_{r};t)}{\partial I^{\prime|\alpha_{2}|+1}\partial\varphi^{\prime|\beta_{2}|}}\frac{\partial \Lambda_{r}}{\partial I^{\prime}}s\frac{\partial \Xi_{r}}{\partial\varphi^{\prime}}\\
&\quad+\frac{\partial^{|\beta_{2}|+1}f_{r-1}^{1}(I^{\prime}+\Xi_{r},\varphi+
\Lambda_{r};t)}{\partial I^{\prime}\partial\varphi^{\prime|\beta_{2}|}}s\frac{\partial^{|\alpha_{2}|+1} \Xi_{r}}{\partial I^{\prime|\alpha_{2}|}\partial\varphi^{\prime}}\\
&\quad+\frac{\partial^{|\beta_{2}|+1}f_{r-1}^{2}(I^{\prime}+\Xi_{r},\varphi+
\Lambda_{r};t)}{\partial I^{\prime|}\partial\varphi^{\prime|\beta_{2}|}}\frac{\partial^{|\alpha_{2}|+1}\Lambda_{r}}{\partial I^{\prime|\alpha_{2}|} \partial\varphi^{\prime}}\Bigg]\times D_{I^{\prime}}^{\alpha-\alpha_{2}}D_{\varphi^{\prime}}^{\beta-\beta_{2}}\Xi_{r}\\
&\quad\sim  O(t^{2^{r}})+O(t^{2^{r}+1})+\cdots+O(t^{2^{r+1}+2^{r-1}+\cdots+2^{1}})
\end{align*}
from \eqref{v1},\eqref{a2} and the fact $h_{r-1}\in\mathcal{G}^{\rho,\rho+1}$, $f_{r-1}\in\mathcal{G}^{\rho,\rho+1,\rho+1}$.
We thus attain the desired outcome. Here, we provide a notation explanation:
$\partial^{i,j},\ i,j=1,2$ stands for taking the partial derivative with respect to the first variable, and then with respect to the second variable.
\end{proof}
It is worth noting that we have only considered finite transformations, so there is no need to address convergence issues.
\subsection{Calculation of $K_{0}(I;t)$}
We consider the perturbed Hamiltonian $H(\theta,I;t)\in\mathcal{G}^{\rho,\rho,1}_{L_{1},L_{2},L_{2}}(\mathbb{T}^{d}\times D_{0}\times(-1,1))$, that is, $H$ is analytic with respect to $t$. Hence, we can write
\begin{align*}
H(\varphi,I;t)=h_{0}(I)+f_{0}(\varphi,I;t)
\end{align*}
with
\begin{align*}
f_{0}(\varphi,I;t)&=H(\varphi,I;t)-h_{0}(I)\\
&=t\cdot\partial_{t}H(\theta,I;0)+\frac{1}{2!}t^{2}\cdot\partial_{t}^{2}H(\theta,I;0)+\cdots\\
&\quad+\frac{1}{(d-1)!}t^{d-1}\cdot\partial_{t}^{d-1}H(\theta,I;0)+O(t^{d}).
\end{align*}
The notation $\partial_{t}^{(d)}$ represents the $d-$st order derivative with respect to $t$.

Based on the above discussion, we now turn to our setting and aim to find a symplectic transformation $\chi_{0}$  parameterized by
$t$ such that
\begin{align*}
H^{1}(\varphi,I;t)=(H\circ\chi_{0})(\varphi,I;t)=h_{0}(I)+t\cdot(2\pi)^{-d}\int_{\mathbb{T}^{d}}\partial_{t}H(\varphi,I;0)d\varphi+f_{1}(\varphi,I;t),
\end{align*}
where the error term $f_{1}(\varphi,I;t)=O(t^{2})+O(t^{3})+O(t^{4})$. The integrable Hamiltonian is
\begin{align*}
H^{1}_{0}(I)=h_{0}(I)+t\cdot(2\pi)^{-d}\int_{\mathbb{T}^{d}}\partial_{t}H(\varphi,I;0)d\varphi.
\end{align*}
We repeatedly use the KAM iteration process $r-1\in \mathbb{Z}_{+}$ times. We can get
\begin{align*}
H^{r-1}(\varphi,I;t)&=(H^{r-2}\circ\chi_{r-2})(\varphi,I;t)\\
&=h_{0}(I)+t\cdot(2\pi)^{-d}\int_{\mathbb{T}^{d}}\partial_{t}H(\varphi,I;0)d\varphi+\cdots\\
&\quad +t^{2^{r-1}+2^{r-3}+\cdots+2^{1}}\cdot(2\pi)^{-d}\int_{\mathbb{T}^{d}}
\partial_{t}^{(2^{r-1}+2^{r-3}+\cdots+2^{1})}\tilde{H}(\varphi,I;0)d\varphi\\
&\quad+f_{r-1}(\varphi,I;t)
\end{align*}
with $f_{r-1}(\varphi,I;t)=O(t^{2^{r-1}})+O(t^{2^{r-1}+1})+\cdots+O(t^{2^{r}+2^{r-2}+\cdots+2^{1}})$ for $2^{r}+2^{r-2}+\cdots+2^{1}=d-1$. The reason that the function $\partial_{t}^{(2^{r-1}+2^{r-3}+\cdots+2^{1})}\tilde{H}$ appears here is that certain powers of $t$ repeat, so this is not the function $\partial_{t}^{(2^{r-1}+2^{r-3}+\cdots+2^{1})}H$.
Consider this transformed Hamiltonian $H^{r-1}$ as being a small perturbation of the following integrable Hamiltonian
\begin{align*}
H_{0}^{r-1}(I;t)&=h_{0}(I)+t\cdot(2\pi)^{-d}\int_{\mathbb{T}^{d}}\partial_{t}H(\varphi,I;0)d\varphi+t^{2}\cdot(2\pi)^{-d}\int_{\mathbb{T}^{d}}
\partial_{t}^{(2)}H(\varphi,I;0)d\varphi\\
&\quad+\cdots +t^{2^{r-1}+2^{r-3}+\cdots+2^{1}}(2\pi)^{-d}\int_{\mathbb{T}^{d}}\partial_{t}^{(2^{r-1}+2^{r-3}+\cdots+2^{1})}\tilde{H}
(\varphi,I;0)d\varphi.
\end{align*}
It is easy to check that $K(I;t)\in\mathcal{G}^{\rho,\rho^{\prime}}$, where $\rho^{\prime}=\rho+1$. Moreover, we perform the transform one more again, it can be demonstrated that
\begin{align}\label{r1}
K(I;t)=h_{0}(I)&+t\cdot(2\pi)^{-d}\int_{\mathbb{T}^{d}}\partial_{t}H(\varphi,I;0)d\varphi+t^{2}\cdot(2\pi)^{-d}\int_{\mathbb{T}^{d}}
\partial_{t}^{(2)}H(\theta,I;0)d\varphi+\cdots\\ \notag
&+t^{d-1} (2\pi)^{-d}\int_{\mathbb{T}^{d}}\partial_{t}^{(d-1)}H(\varphi,I;0)d\varphi+O(t^{d}).
\end{align}

\section{Semiclassical scarring}
In this section, we mainly focus on preparing for proving Theorem \ref{ai}, but this necessitates the identification of  a family of quasimodes that satisfy the condition stipulated in Corollary \ref{br}.  After getting  such quasimodes, we investigate the behaviour of the eigenfunctions corresponding to the eigenvalues that are within a distance of $O(h^{\gamma})$ from the quasieigenvalues, in the semiclassical limit $h\rightarrow0$, for $h-$pseudodifferential operator (denoted as h-PDO) $P_{h}(t)$. As described in the introduction, when $\nu$ has positive mass on a set $S$, we say that the  sequence of eigenfunction scars, or concentrates, at $S$.
\par
In our work, if we can find a family of quasimodes, then we can also show that in dimension $d$, for almost all $t\in[0,t_{0}]$ and for a full measure set of invariant tori $\Lambda_{\omega}$, there exist semiclassical measures for $P_{h}(t)$ with positive mass on $\Lambda_{\omega}$. For a system with $2d$ degrees of freedom, the energy surfaces have dimension $2d-1$ and Lagrangian tori have dimension $d$, this thus implies the existence of sequences of eigenfunctions that scar on $\Lambda_{\omega}$.
\subsection{Separation of quasieigenvalue}
To analyze spectral concentration in this section, we draw on the idea of using the spectral flow of a 1-parameter family of operators, which is discussed in Hassell's paper (\cite{MR2630052}). First, for any the non-degenerate completely integrable Hamiltonian $h_{0}(I)=H(\varphi,I;0)$, we consider a one-parameter family of the perturbed Hamiltonian
\begin{align*}
H(\varphi,I;t)\in \mathcal{G}^{\rho,\rho,1}(\mathbb{T}^{d}\times D\times(-1,1)).
\end{align*}
We choose $\kappa$ sufficiently small so that the set of non-resonant frequency set $\Omega_{\sigma}$ has positive measure. And without loss of generality, we assume that $D$ is convex by sharking if necessary.
\par
Based on Corollary \ref{br}, we have derived a $\mathcal{G}^{\rho}-$quasimode family $\mathcal{Q}$ of $P_{h}(t)$ with an index set $\mathcal{M}_{h}(t)$ in view of quantum Birkhoff normal form. We take quasi-eigenvalue as the form $\lambda(t;h)=K^{0}(h(m+\vartheta/4);t,h)$, where $I_{m}=h(m+\vartheta/4),m\in \mathcal{M}_{h}(t)$.
\begin{remark}
We will write $I_{m}=h(m+\vartheta/4)$, for $m\in \mathcal{M}_{h}(t)$ or not, so a fixed $I_{m}\in h(\mathbb{Z}^{d}+\vartheta/4),\ m\in \mathcal{M}_{h}(t)$ will only be in $E_{\kappa}(t)$ for $O(h)$-sized intervals as $t$ varies.
\end{remark}
For any two distinct $I_{m}, I_{m^{\prime}}$, we have:
\begin{proposition}\label{aj}
Suppose that
\begin{align*}
 \lim_{h\rightarrow 0}\frac{C_{1}\kappa h^{-1/2}}{( \Delta^{-1}( C_{1}h^{-1/2}))^{2}}=\infty.
\end{align*}
Then, for all distinct $m,\ m^{\prime}\in \mathbb{Z}^{d}$ with $I_{m},\ I_{m^{\prime}}\in D$ such that
\begin{align*}
|I_{m}-I_{m^{\prime}}|\leq h \Delta^{-1}( C_{1}h^{-1/2}),
\end{align*}
and $m\in \mathcal{M}_{h}(t)$, we have
\begin{align*}
| \mu_{m}-\mu_{m^{\prime}}|\geq C_{2}h^{3/2},
\end{align*}
where the constants $C_{1},C_{2}$ depend on the choice of perturbation $H$ and on the non-resonant constant $\kappa$ but independent of $t$ and $h$.
\end{proposition}
\begin{proof}
Firstly, we prove the first estimate:
\begin{align}\label{q}
|I_{m}-I_{m^{\prime}}|\leq h \Delta^{-1}( C_{1}h^{-1/2})\quad \text{if and only if}\quad |m-m^{\prime}|\leq \Delta^{-1}( C_{1}h^{-1/2}).
\end{align}
Next, we already know that $K^{0}$ has a semiclassical expansion, that is,
\begin{align*}
K^{0}(I;t,h)=\sum_{0\leq j\leq\eta h^{-1/\rho}}K_{j}(I;t)h^{j},
\end{align*}
and following from the leading order term in the semiclassical expansion of $K^{0}$, we have
\begin{align*}
|\mu_{m}-\mu_{m^{\prime}}|\geq |K_{0}(I_{m},t)-K_{0}(I_{m^{\prime}},t)|+O(h^{2})
\end{align*}
uniformly for $t<t_{0}$. The function $K_{0}(I_{m^{\prime}},t)$  undergoes a Taylor expansion with respect to the variable $I$ at $I_{m}$ to second order:
\begin{align}\label{r}
K_{0}(I_{m^{\prime}},t)=K_{0}(I_{m},t)+\langle \nabla_{I}K_{0}(I_{m},t),(I_{n}-I_{m})\rangle+\langle \nabla_{I}^{2}K_{0}(\hat{I},t)(I_{m^{\prime}}-I_{m}),(I_{m^{\prime}}-I_{m})\rangle
\end{align}
for some $\hat{I}$ on the line segment between $I_{m^{\prime}}$ and $I_{m}$ in components.
\par
Since $m\in \mathcal{M}_{h}(t)$, we also have
\begin{align*}
|\nabla_{I}K_{0}(I_{m},t)-\nabla_{I}K_{0}(I_{\omega},t)|=O(h)
\end{align*}
uniformly to $t<t_{0}$, where $I_{\omega}$ is some non-resonant action corresponding to a non-resonant frequency $\omega\in \Omega_{\sigma}$. Bring this estimate into \eqref{r}, and using the fact $\omega=\nabla_{I}K_{0}(I_{\omega},t),\ I_{m^{\prime}}-I_{m}=h(m^{\prime}-m)$, we get
\begin{align}\label{s}
|\mu_{m}-\mu_{m^{\prime}}|&\geq |K_{0}(I_{m},t)-K_{0}(I_{m^{\prime}},t)|\\ \notag
& =h\nabla_{I}K_{0}(I_{\omega},t)\cdot(m-m^{\prime})+O(h^{2}|m-m^{\prime}|^{2})+O(h^{2}|m-m^{\prime}|)\\ \notag
&  =h\nabla_{I}K_{0}(I_{\omega},t)\cdot(m-m^{\prime})+O(h^{2}|m-m^{\prime}|^{2})\\ \notag
& =h\omega\cdot(m-m^{\prime})+O(h^{2}|m-m^{\prime}|^{2})\\ \notag
& \geq \frac{h\kappa}{\Delta(|m-m^{\prime}|)}+O(h^{2}|m-m^{\prime}|^{2})\\ \notag
& \geq \kappa C_{1} h^{3/2}+O((h \Delta^{-1}( C_{1}h^{-1/2}))^{2}).
\end{align}
For the assumption:
\begin{align}
 \lim_{h\rightarrow 0}\frac{C_{1}\kappa h^{-1/2}}{( \Delta^{-1}( C_{1}h^{-1/2}))^{2}}=\infty,
\end{align}
we demand that  $\Delta^{-1}( C_{1}h^{-1/2}))\rightarrow\infty,$ as $h\rightarrow0$, and slower than $h^{1/4}\rightarrow0$. We can infer that $h^{3/2}$ is lower-order term. This claim is proved upon choosing $C_{1}$ appropriately.
\end{proof}
\begin{remark}
It is worth noting that the inverse of Proposition \ref{aj} is true. In a word, if two distinct quasi-eigenvalues $\mu_{m},\ \mu_{m^{\prime}}$ are very close (even less than $C_{2}h^{3/2}$ apart), then their actions $I_{m},\ I_{m^{\prime}}$ are at least $h\Delta^{-1}(C_{1}h^{-1/2})$ distance apart.
\end{remark}
From the geometric assumption \eqref{a4} on the perturbation family $H(\theta,I;t)$, we make a claim that $K^{0}$(the integral part of Quantum Birkhoff normal form) and $\partial_{t} K^{0},\cdots,\partial_{t}^{(d-1)}K^{0}$ locally form coordinates in $D$ for all $t<t_{0}$ and $h<h_{0}$. This is because  in dimension $d$, under the assumptions in Theorem \ref{ai}, the level set of these $d$ quantities intersect transversally and form a coordinate system for the action space $D$.  More precisely, we have the following.
\begin{proposition}\label{am}
There exist $t_{0},h_{0}>0$ such that for all $t<t_{0}$ and $h<h_{0}$,
\begin{align*}
\eta:I\rightarrow\left(K^{0}(I;t,h),\partial_{t}K^{0}(I;t,h),\cdots, \partial_{t}^{(d-1)}K^{0}(I;t,h)\right)
\end{align*}
is a local diffeomorphism. Moreover, we have
\begin{align*}
G_{1}|\eta(I_{1})-\eta(I_{2})|\leq|I_{1}-I_{2}|\leq G_{2}|(I_{1})-\eta(I_{2})|
\end{align*}
for some positive constants $G_{1},G_{2}$ that depend on our choice of perturbation $H$ but are uniform in $t$ and $h$.
\end{proposition}
\begin{proof}
we have known that
\begin{align*}
K^{0}(I;t,h)=\sum_{j\leq \gamma}K_{j}(I;t)h^{j}
\end{align*}
with each $K_{j}$ smooth and \eqref{r1}. We have
\begin{align*}
K^{0}(I;t,h)=K_{0}(I;0)+O(h)=K_{0}(I;t)+O(h)+O(t).
\end{align*}
It follows that
\begin{align*}
\partial_{t}K^{0}(I;0,h)=(2\pi)^{-d}\int_{\mathbb{T}^{d}}\partial_{t}H(\theta,I;0)d\theta+O(h)
\end{align*}
uniformly in $I$. Hence
\begin{align*}
\partial_{t}K^{0}(I;t,h)=(2\pi)^{-d}\int_{\mathbb{T}^{d}}\partial_{t}H(\theta,I;0)d\theta+O(h)+O(t).
\end{align*}
Similarly,
\begin{align*}
\partial_{t}^{(2)}K^{0}(I;t,h)=(2\pi)^{-d}\int_{\mathbb{T}^{d}}\partial_{t}^{(2)}H(\theta,I;0)d\theta+O(h)+O(t),
\end{align*}
$$\vdots$$
\begin{align*}
\partial_{t}^{(d-1)}K^{0}(I;t,h)=(2\pi)^{-d}\int_{\mathbb{T}^{d}}\partial_{t}^{(d-1)}H(\theta,I;0)d\theta+O(h)+O(t).
\end{align*}
From the geometric assumption \eqref{a4}, the claim is proved.
\end{proof}
By applying Proposition \ref{am}, we obtain a lower bound on the difference of speeds $\partial_{t}(\mu_{m}-\mu_{m^{\prime}})$. Consequently, this implies that two quantities diverge rapidly as $t$ evolves. This behavior can be quantified as follows:
\begin{proposition}\label{ak}
Choose any $\gamma>7/4$. Suppose that $h<h_{0},\ m,\ m^{\prime}\in \mathbb{Z}^{d}$ are distinct, and  $t_{*}\in (0,t_{0})$ are fixed with $I_{m},\ I_{m^{\prime}}\in D,\ m\in \mathcal{M}_{h}(t)$ and
\begin{align}\label{t}
|\ \mu_{m}(t_{*},h)-\mu_{m^{\prime}}(t_{*},h)|<h^{\gamma}<h^{7/4}.
\end{align}
Set
\begin{align*}
\mathcal{C}_{m,m^{\prime}}(h)=\{t\in (0,t_{0}):|\ \mu_{m}(t,h)-\mu_{m^{\prime}}(t,h)|<h^{\gamma}\}.
\end{align*}
Then there exist positive constants $\tilde{C}_{1},\tilde{C}_{2}$ which depend on the constants $C_{1},\ C_{2}$ from Proposition \ref{aj} as well as the geometric constants $G_{1},\ G_{2}$ from Proposition \ref{am} such that
\begin{align}\label{bf}
\frac{meas\left([t_{*}-\tilde{C}_{1}h^{3/4},t_{*}+\tilde{C}_{1}h^{3/4}]\cap\mathcal{C}_{m,n}\right)}{h^{3/4}}<\tilde{C}_{2}h^{\gamma-7/4}(\Delta^{-1}(C_{1}h^{-1/2}))^{-1}.
\end{align}
Moreover, when $\gamma>7/4+d$, we also have the following estimate
\begin{align*}
meas\left(\left\{t\in(0,t_{0}):m\ \in\ \mathcal{M}_{h}(t)\ and\ |\mu_{m}-\mu_{m^{\prime}}|  < h^{\gamma}\ for\ all\ m^{\prime}\ \neq\ m, I_{m^{\prime}}\ \in\ D\right\}\right)\\
=O(h^{\gamma-7/4-d}(\Delta^{-1}(C_{1}h^{-1/2}))^{-1}).
\end{align*}
\end{proposition}
\begin{proof}
Applying the inverse of Proposition \ref{aj}, we have $|I_{m}-I_{m^{\prime}}|\geq h\Delta^{-1}(C_{1}h^{-1/2})$ and as an application of Proposition \ref{am}, we obtain
\begin{align*}
|\partial_{t}\mu_{m}(t;h)-\partial_{t}\mu_{m^{\prime}}(t;h)| = |\partial_{t}K^{0}(I_{m},t;h)-\partial_{t}K^{0}(I_{m^{\prime}},t;h)\geq C\ h\Delta^{-1}(C_{1}h^{-1/2}),
\end{align*}
where $C$ relies on $C_{1},\ C_{2}$ and the geometric constants $G_{i},i=1,2$. We do Taylor's expansion of the time parameter $t$ to the second order, then we have
\begin{align*}
\mu_{m}(t;h)=K^{0}(I_{m},t;h)=K^{0}(I_{m},t_{*};h)+(t-t_{*})\partial_{t}(K^{0}(I_{m},t;h))+O(|t-t_{*}|^{2}).
\end{align*}
In a similar way, by Taylor's expansion, we also have
\begin{align*}
\mu_{m^{\prime}}(t;h)=K^{0}(I_{m^{\prime}},t;h)=K^{0}(I_{m^{\prime}},t_{*};h)+(t-t_{*})\partial_{t}(K^{0}(I_{m^{\prime}},t;h))+O(|t-t_{*}|^{2}),
\end{align*}
and their error terms are uniform in $h$ and $m$. Make the difference between the above two formulas, and use formula \eqref{t}:
\begin{align*}
|\ \mu_{m}(t;h)-\mu_{m^{\prime}}(t;h)|&=(t-t_{*})|\partial_{t}(K^{0}(I_{m},t;h))-\partial_{t}(K^{0}(I_{m^{\prime}},t;h))|\\
&\quad+O(h^{\gamma})+O(|t-t_{*}|^{2}).
\end{align*}
To get the results we expect, we choose $\tilde{C}_{1}$, such that the linear term for $|t-t_{*}|<\tilde{C}_{1}h^{3/4}$ can controlled the quadratic term $O(|t-t_{*}|^{2})$. Moreover, due to $\gamma>7/4$, $h^{\gamma}$ is controlled by the others for sufficiently small $h$. It follows that
\begin{align*}
|\ \mu_{m}(t;h)-\mu_{m^{\prime}}(t;h)| & \geq \frac{1}{2}(t-t_{*})|\partial_{t}(K^{0}(I_{m},t;h))-\partial_{t}(K^{0}(I_{m^{\prime}},t;h))| \\
&\geq \frac{1}{2}C(t-t_{*})\tilde{C}_{1}h\Delta^{-1}(C_{1}h^{-1/2})
\end{align*}
for $t\in [t_{*}-\tilde{C}_{1}h^{3/4},t_{*}+\tilde{C}_{1}h^{3/4}]$. Hence, we have
\begin{align*}
t\in\mathcal{C}_{m,m^{\prime}} & \Longleftrightarrow \frac{1}{2}C(t-t_{*})\tilde{C}_{1}h\Delta^{-1}(C_{1}h^{-1/2})\leq h^{\gamma}\\
& \Longleftrightarrow|t-t_{*}|\leq C^{-1}h^{\gamma-1}(\Delta^{-1}(C_{1}h^{-1/2}))^{-1}.
\end{align*}
This yields the estimate of this Proposition, where the constants $\tilde{C}_{i}$ depend on the original Hamiltonian, the  perturbation, $\kappa$, and $G_{i}$, but not on $t$ or $h$. Finally, set
\begin{align*}
A_{m}=\{t\in(0,t_{0}):m\ \in\ \mathcal{M}_{h}(t)\ and\ |\ \mu_{m}-\mu_{m^{\prime}}| < h^{\gamma}\ for\ all\ m^{\prime}\ \neq\ m, I_{m^{\prime}}\ \in\ D\}.
\end{align*}
By recalling a measure-theoretic lemma \ref{an} and using estimate \eqref{bf}, we get
\begin{align*}
meas(\{t\in(0,t_{0}):m\ \in\ \mathcal{M}_{h}(t)\ and\ |\ \mu_{m}-\mu_{m^{\prime}}| < h^{\gamma}\})\\=O(h^{\gamma-7/4}(\Delta^{-1}(C_{1}h^{-1/2}))^{-1}),
\end{align*}
then by summing over all such $m^{\prime}$, $\{m^{\prime}\in\mathbb{Z}^{d}:\ I_{m^{\prime}}\in D\}$, we have
\begin{align*}
meas(A_{m})&=h^{-d}\cdot meas(D)\cdot O(h^{\gamma-7/4}\Delta^{-1}(C_{1}h^{-1/2}))\\&=O(h^{\gamma-7/4-d}(\Delta^{-1}(C_{1}h^{-1/2}))^{-1}).
\end{align*}
This completes the proof.
\end{proof}
From Proposition \ref{ak}, we can draw a conclusion that $\forall t\in(0,t_{0})$, we have $|\mu_{m}-\mu_{m^{\prime}}|>h^{\gamma}$, for $m\in \mathcal{M}_{h}(t)$ with $m^{\prime}\in\mathbb{Z}^{d}, m^{\prime}\neq m,\ I_{m},I_{m^{\prime}}\in D$. And, if we introduce so-called energy window $\mathcal{W}_{m}(h):=[\mu_{m}(h)-\frac{h^{\gamma}}{3},\mu_{m}(h)+\frac{h^{\gamma}}{3}]$, then $\mu_{m}$ is only quasieigenvalue in $\mathcal{W}_{m}(h),\ m\in \mathcal{M}_{h}(t),\ t\in(0,t_{0})$.
\par
\begin{remark}
In the definition of $A_{m}$, it is not required that $m^{\prime}\in\mathcal{M}_{h}(t)$.
\end{remark}
To formulate this result more precisely, we introduce a new notation. We consider a new set:
\begin{align*}
N(t;h)=\sharp \mathcal{M}_{h}(t)=\sharp \{m\in\mathbb{Z}^{d}:I_{m}\in D\ and\ t\in A_{m}\}.
\end{align*}
\begin{proposition}\label{ap}
Let $N$ be defined as above. Then for $\gamma>7/4+2d$, the set
\begin{align*}
\mathfrak{g}:=\{t\in(0,t_{0}),\exists\ sequence\ h_{j}\ \rightarrow\ 0\ such\ that\ N(t;h_{j})=0\ for\ all\ j\}
\end{align*}
has full measure in $(0,t_{0})$.
\end{proposition}
The detailed proof can be found in Reference \cite{MR4404789}.
\subsection{Positive mass}
The aim of this subsection is to study the number of eigenvalue $E_{k}(h_{j})$ lying in the window $\mathcal{W}_{m}(h):=[\mu_{m}(h)-\frac{h^{\gamma}}{3},\mu_{m}(h)+\frac{h^{\gamma}}{3}],m\in\mathcal{M}_{h}(t)$. According to the Proposition \ref{ak}, the energy windows $\mathcal{W}_{m}(h_{j})$ are disjoint. Having established that $\mathfrak{g}$ is of full measure in Proposition \ref{ap} provided $\gamma>7/4+2d$, we now fix $\gamma>7/4+2d$. For simplicity, we will no longer writing time $t$.
\par
Let
\begin{align*}
\mathcal{N}_{m}(h)=\sharp\{E_{k}(h)\in \mathcal{W}_{m}(h)\}.
\end{align*}
According to $Weyl's$ law, the total number of eigenvalue $E_{k}(h_{j})$ lying in the energy band $[a,b]$ is asymptotically equal to
$(2\pi h)^{-d}meas(p^{-1}([a,b]))$, where $p=\sigma(P_{h})$, meanwhile, the number of quasi-modes in our local patch $\mathbb{T}^{d}\times D$ is asymptotically equivalent to $ (2\pi h_{j})^{-d}meas(E_{\kappa})$.
\par
Fixing $\lambda>1$, we define
\begin{align}\label{w}
\tilde{\mathcal{M}}_{h}(\lambda):=\left\{m\in\mathcal{M}_{h}:\mathcal{N}_{m}(h)<\lambda \frac{meas(p^{-1}([a,b]))}{meas(E_{\kappa})}\right\}.
\end{align}
\begin{lemma}\label{ar}
For fixed $\gamma>7/4+2d$ and each $m\in \tilde{\mathcal{M}}_{h_{j}}(\lambda)$, there exists an $L^{2}$-normalised eigenfunction $u_{k_{j}}(h_{j})$ with eigenvalue $E_{k_{j}}(h_{j})\in [\mu_{m}-h^{\gamma}_{j}/3,\mu_{m}+h^{\gamma}_{j}/3]$ such that
\begin{align}
\max_{|E_{k_{j}}-\mu_{m}|\leq h^{\gamma}_{j}}|\langle u_{k_{j}},v_{m}\rangle|\geq\frac{1-o(1)}{\lambda }\cdot\frac{meas(E_{\kappa})}{meas(p^{-1}([a,b])}.
\end{align}
\end{lemma}
The proof of this lemma can be found in \cite{MR4404789}.
We introduce the notation
\[
\begin{aligned}
\tilde{\mathcal{J}}_{h_{j}}(\lambda)&=\{h_{j}(m+\vartheta/4): m\in\tilde{\mathcal{M}}_{h_{j}}(\lambda)\},\\
\mathcal{J}_{h_{j}}(\lambda)&=\{h_{j}(m+\vartheta/4): m\in\mathcal{M}_{h_{j}}\}.
\end{aligned}
\]
We will prove that the distance between most resonant actions $E_{\kappa}$ of KAM tori and actions in $\tilde{\mathcal{M}}_{h_{j}}(\lambda)$ is $O(h_{j})$ for all sufficiently large $j$. This shows that the set $\tilde{\mathcal{M}}_{h}(\lambda)$ consisting in index corresponding the concentrating quasimodes associated to such torus actions is contained in $\mathcal{M}_{h_{j}}$, that is $\tilde{\mathcal{M}}_{h}(\lambda)\subset \mathcal{M}_{h_{j}}$.
\begin{proposition}\label{aq}
Let $\tilde{\mathcal{J}}_{h_{j}}$ be defined as above. Then we have
\begin{align}\label{y}
\frac{meas(\{I\in E_{\kappa}:dist(I,\tilde{\mathcal{J}}_{h_{j}}(\lambda))<Lh_{j}\})}{meas(E_{\kappa})}\geq1-\frac{L^{d}}{\pi^\frac{d}{2}\lambda}
\end{align}
for all sufficiently large $j$.
\end{proposition}
\begin{proof}
To prove estimate \eqref{y}, we have to use other estimates. The Pigeonhole Principle shows that for large $\lambda,\ \mathcal{N}_{m}(h)$ is only rarely larger than $\lambda\cdot \frac{meas(p^{-1}([a,b]))}{meas(E_{\kappa})}$. Indeed, the disjointness of the $\mathcal{W}_{m}(h_{j})$ implies
\begin{align*}
\limsup_{j\rightarrow\infty}h_{j}^{d}\ \sharp(\mathcal{M}_{h_{j}}\setminus\tilde{\mathcal{M}}_{h_{j}})\cdot\lambda \frac{meas(p^{-1}([a,b]))}{meas(E_{\kappa})} \leq \frac{meas(p^{-1}([a,b]))}{4\pi^{2}}.
\end{align*}
Due to $h_{j}^{d}\ \sharp\mathcal{M}_{h_{j}}\rightarrow (2\pi)^{-d}meas(E_{\kappa})$ as $j\rightarrow\infty$, we come to the conclusion
\begin{align}\label{x}
\sharp(\mathcal{M}_{h_{j}}\backslash\tilde{\mathcal{M}_{h_{j}}})<\frac{2}{\lambda}\cdot (2\pi h_{j})^{-d}meas(E_{\kappa}),
\end{align}
as desired . Then by formula \eqref{x}, we obtain
\begin{align*}
& meas(\{I\in E_{\kappa}:dist(I,\tilde{\mathcal{J}}_{h_{j}})(\lambda)<Lh_{j}\})\\ \notag
 &\quad\geq  1-meas(\{I\in E_{\kappa}:dist(I,\mathcal{J}_{h_{j}}\backslash\tilde{\mathcal{J}}_{h_{j}})(\lambda)<Lh_{j}\})\\ \notag
&\quad\geq  1-\frac{1}{meas(E_{\kappa})}\cdot\sharp(\mathcal{M}_{h_{j}}\backslash\tilde{\mathcal{M}}_{h_{j}})(\lambda))\cdot\pi^{\frac{d}{2}} (Lh_{j})^{d}\\ \notag
&\quad\geq  1-\frac{1}{meas(E_{\kappa})}\cdot\frac{2}{\lambda}(2\pi h_{j})^{-2} meas(E_{\kappa})\cdot\pi^{\frac{d}{2}}(Lh_{j})^{d}\\ \notag
&\quad=  1-\frac{L^{d}}{\pi^{\frac{d}{2}}\lambda}
\end{align*}
for all sufficiently large $j$.
\end{proof}
Due to the proof process of Proposition \ref{aq}, we could make a claim:
\begin{remark}
After exactly calculation, we get
\begin{align*}
\frac{\sharp\tilde{\mathcal{M}}_{h_{j}}}{\sharp\mathcal{M}_{h_{j}}}=1-\frac{\sharp(\mathcal{M}_{h_{j}}\setminus\tilde{\mathcal{M}}_{h_{j}})}{\sharp\mathcal{M}_{h_{j}}}>1-
\frac{2}{\lambda}
\end{align*}
for each sufficiently large $j$ by \eqref{x}. In other words, the proportion of $O(h^{\gamma})$-sized energy windows associated to actions in $\mathcal{M}_{h}$ containing at most $\lambda R$ eigenvalue is at least $1-\frac{2}{\lambda}$, when $\lambda>2$.
 \end{remark}
\subsection{Proof of Theorem \ref{ai}}
The key ingredient of proving Theorem \ref{ai} is Proposition \ref{aq} and Lemma \ref{ar}. The idea of this proof is similar to \cite{MR4404789}.
\begin{proof}
Take the subset of $E_{\kappa}$ in \eqref{y} denoted as  $E_{\kappa,j}(\lambda)$, and introduce a new notation as follows:
\begin{align*}
\tilde{E}_{\kappa}(\lambda):=\bigcap_{l=1}^{\infty}\bigcup_{j=l}^{\infty}E_{\kappa,j}(\lambda).
\end{align*}
Then, form Proposition \ref{aq}, we can conclude that $\tilde{E}_{\kappa}(\lambda)$ has measure at least $1-O(\lambda^{-1})$, and  for any $I\in \tilde{E}_{\kappa}(\lambda)$, we may derive that
\begin{align*}
dist(I,\tilde{\mathcal{J}}_{h_{j}}(\lambda))<Lh_{j}
\end{align*}
with infinitely many $j$. For each $I\in \tilde{E}_{\kappa}(\lambda)$ and each such $j$, we choose such an action in $\tilde{\mathcal{J}}_{h_{j}}(\lambda)$, and an associated quasimode $v_{m_{j}}$ for $P_{h_{j}}$, in order to obtain a sequence of quasimodes that concentrates completely on the torus $\Lambda_{\omega}=\{I_{\omega}\}\times\mathbb{T}^{d}$.
\par
For this sequence, we can look for a corresponding sequence of eigenfunction $u_{k_{j}}$ for  $P_{h_{j}}$ by making use of Lemma \ref{ar} such that
\begin{align}\label{u1}
|\langle u_{k_{j}}(h_{j}),v_{m_{j}}(h_{j})\rangle|>\frac{1}{2\lambda}\cdot\frac{meas(E_{\kappa})}{meas(p^{-1}([a,b]))}
\end{align}
for all sufficiently large $j$.
\par
We now claim that sequence $u_{k_{j}}(h_{j})$ scars on the torus $\Lambda_{\omega}$. This is because we can take an arbitrary semiclassical pseudo-differential operator $A_{h}$ with the symbol equal to 1 and compactly supported  around the torus $\Lambda_{\omega}$, and estimate
\begin{align*}
\langle A_{h_{j}}^{2}u_{k_{j}}(h_{j}),u_{k_{j}}(h_{j})\rangle & =\|A_{h_{j}}u_{k_{j}}(h_{j})\|^{2}\\ \notag
& \geq|\langle A_{h_{j}}u_{k_{j}}(h_{j}),v_{m_{j}}(h_{j})\rangle|^{2}\\ \notag
& =|\langle u_{k_{j}}(h_{j}),v_{m_{j}}(h_{j})\rangle+\langle u_{k_{j}}(h_{j}),(A_{h_{j}}-Id)v_{m_{j}}(h_{j})\rangle|^{2}\\ \notag
& >\frac{1}{(2\lambda )^{2}}\cdot\left(\frac{meas(E_{\kappa})}{meas(p^{-1}([a,b]))}\right)^{2}>0
\end{align*}
for all sufficiently large $j$, by \eqref{u1}, and the fact that $v_{m_{j}}$ is a sequence of quasimodes concentrating completely on the torus $\Lambda_{\omega}=\{I_{\omega}\}\times\mathbb{T}^{d}$.
\par
Now, let $\nu$ be a semiclassical measure associated to a subsequence of the $u_{k_{j}}(h_{j})$, then we can conclude that
\begin{align*}
\int\sigma(A)d\nu
\end{align*}
is bounded below by $\frac{1}{(2\lambda )^{2}}\cdot\left(\frac{meas(E_{\kappa})}{meas(p^{-1}([a,b]))}\right)^{2}$. By taking $A$ to have shrinking support around $\Lambda_{\omega}$, we see that $\nu$ has positive mass strictly greater than $\frac{1}{(2\lambda )^{2}}\cdot\left(\frac{meas(E_{\kappa})}{meas(p^{-1}([a,b]))}\right)^{2}$ on $\Lambda_{\omega}$.
\par
Applying this argument with $\lambda\rightarrow\infty$, we establish the existence of such semiclassical measure for almost all $I_{\omega}\in E_{\kappa}$ and the proof is finished.
\end{proof}
One can apply Theorem \ref{ai} with $\kappa\rightarrow0$ to obtain the following corollary in \cite{MR4404789}:
\begin{corollary}
Under the same assumption as in Theorem \ref{ai}, for almost all non-resonant frequencies $\omega\in\bigcup_{\kappa>0}E_{\kappa}$, there exists a semiclassical measure associated to the eigenfunctions of $P_{h}(t)$ that has positive mass on $\Lambda_{\omega}$.
\end{corollary}
\section{Example}
In this section, we give an example.
\begin{align*}
H(\varphi,I;t)=\frac{1}{2}\sum_{i}I_{i}^{2}+tf_{1}(\varphi)I^{3}+t^{2}f_{2}(\varphi)I^{4}+\cdots+t^{d-1}f_{d-1}(\varphi)I^{d+1},
\end{align*}
where $f_{i}(\varphi),i=1,\cdots,d-1$ are $2\pi-$perodic functions with respect to $\varphi$ and satisfy
\begin{align*}
C_{i}=(2\pi)^{-d}\int_{T^{d}}f_{i}(\varphi)d\varphi\neq0.
\end{align*}
It is easy to calculate
\begin{align*}
\nabla_{I} h_{0}(I) &= I, \\
\nabla_{I} \left( \int_{\mathbb{T}^{d}} \partial_{t} H(\varphi, I; 0) d\varphi \right) &= C_{1} I^{2}, \\
\vdots\\
\nabla_{I} \left( \int_{\mathbb{T}^{d}} \partial_{t}^{(d-1)} H(\varphi, I; 0) d\varphi \right) &= C_{d-1} I^{d}.
\end{align*}

We can conclude that
\begin{align*}
\nabla_{I}h_{0}(I),\ \nabla_{I}\left(\int_{\mathbb{T}^{d}}\partial_{t} H(\varphi,I;0)d\varphi \right),\cdots,\nabla_{I}\left(\int_{\mathbb{T}^{d}}\partial_{t}^{(d-1)} H(\varphi,I;0)d\varphi \right)
\end{align*}
are $d$ linearly independent vectors.
\section{Appendix}
\subsection{Proof of Proposition \ref{j1}}
In this part, to avoid confusion with the variable $t$, we replace $t$ with $\varepsilon$ in the Gevrey norm estimation.
First, we present two essential results as follows:
\begin{lemma}
Let $\omega(I)\in C^{\infty}(E_{\kappa};\mathbb{R}^{d})$ satisfy the following Gevrey type estimates:
\begin{align*}
\left|D^{\alpha}_{I}\omega(I)\right|\leq C_{1}^{|\alpha|}\alpha!^{\rho+1},\ \forall\ I\in E_{\kappa},\alpha\in\mathbb{Z}_{+}^{d},
\end{align*}
\begin{align*}
\left|\langle\omega(I),\gamma\rangle\right|\geq \frac{\kappa}{\Delta(|\gamma|)},\ \forall\ I\in E_{\kappa},\ \gamma\in\mathbb{Z}^{d}\backslash \{0\}.
\end{align*}
Then there exists a positive constant $C_{0}$ depending only on $d,\kappa$ and $C_{1}$, such that
\begin{align}\label{ba}
\left|D^{\alpha}_{I}(\langle\omega(I),\gamma\rangle^{-1})\right|\leq C_{0}^{|\alpha|+1}\alpha!\max_{0\leq j\leq|\alpha|}\left((|\alpha|-j)!^{\rho}|\gamma|^{j}\Delta^{j+1}(|\gamma|)\right)
\end{align}
for any $I\in E_{\kappa},\ 0\neq \gamma\in \mathbb{Z}^{d}$ and $\alpha\in\mathbb{Z}_{+}^{d}$.
\end{lemma}
\begin{lemma}\label{a7}
If
\begin{align*}
\frac{1}{\log\kappa}\int_{T}^{\infty}\frac{\log\Delta(t)}{t^{1+\frac{1}{\sigma}}}{\rm d}t\leq\eta
\end{align*}
for $1<\kappa\leq2$ and $T\geq\varsigma$, then
\begin{align*}
\Gamma_{s}(\eta)\leq e^{\eta(s)T^{\frac{1}{\sigma}}}.
\end{align*}
\end{lemma}
The details of the proof can be found in \cite{a}.

Suppose that $f(\varphi,I)\in C^{\infty}(\mathbb{T}^{n}\times E_{\kappa})$ meets
\begin{align}\label{f}
\left|D_{I}^{\alpha}D_{\varphi}^{\beta}f(\varphi,I)\right|\leq \varepsilon C_{1}^{|\alpha|}C_{2}^{|\beta|}\alpha!^{\mu}\beta!^{\sigma}
\end{align}
for any $I\in E_{\kappa},\ \alpha,\ \beta\in \mathbb{Z}_{+}^{d}$, and $\sigma,\ \mu$ are suitable positive constants.
Assuming that $f$ is an anisotropic Gevrey function as described above and $2\pi-$periodic with respect to $\varphi$, its Fourier coefficients are denoted by $f_{\gamma}(I)$ in the following manner:
\begin{align}\label{bx}
f_{\gamma}(I)=(2\pi)^{-d}\int_{\mathbb{T}^{d}}e^{-{\rm i}\langle \gamma,\varphi\rangle}f(\varphi,I){\rm d}\varphi.
\end{align}
Clearly, $f_{\gamma}(I)$ belongs to $\mathcal{G}^{\mu}$.
Then, this allows for
\begin{align}\label{c}
\left|D_{I}^{\alpha}f_{\gamma}(I)\right|\leq \varepsilon C_{1}^{|\alpha|}\alpha!^{\mu}e^{-C^{-1}|\gamma|^{\frac{1}{\sigma}}}
\end{align}
for $\sigma>1,\ C^{-1}>0$.
\par

Firstly, the Fourier expansions of $f$ is presented below:
\begin{align*}
f(\varphi,I)=\sum_{0\neq \gamma\in\mathbb{Z}^{d}}e^{{\rm i}\langle \gamma,\varphi\rangle}f_{\gamma}(I),
\end{align*}
where
\begin{align*}
f_{\gamma}(I)=(2\pi)^{-d}\int_{\mathbb{T}^{d}}e^{-{\rm i}\langle \gamma,\varphi\rangle}f(\varphi,I){\rm d}\varphi.
\end{align*}
Consider $\psi$ in the same way. Substituting the Fourier expansions of $f$ and $\psi$ into the homological equation, and by comparing the coefficients, we find that
\begin{align*}
\psi_{\gamma}(I)=\langle\omega(I),\gamma\rangle^{-1}f_{\gamma}(I),\quad I\in E_{\kappa},\quad 0\neq \gamma\in\mathbb{Z}^{d}.
\end{align*}
This confirms the existence of the solution $u$. Thus,
\begin{align*}
\psi(\varphi,I)=\sum_{0\neq \gamma\in \mathbb{Z}^{d}}\langle\omega(I),\gamma\rangle^{-1}f_{\gamma}(I)e^{{\rm i}\langle \gamma,\varphi\rangle}.
\end{align*}
We set $W(|\gamma|)=|\gamma|^{-|\beta|}\Delta^{-|\beta|-1}(|\gamma|)$ with $j<|\alpha_{1}|\leq|\beta|$ for the components, as we assume that $\alpha_{1}\leq\beta$.
\par
Take  into account \eqref{ba} and \eqref{c}, then for any $\alpha,\alpha_{1},\beta\in\mathbb{Z}^{d}_{+}$ with $\alpha_{1}\leq\beta$, it holds
\begin{align*}
&\left|D_{I}^{\alpha}D_{\varphi}^{\beta}\psi(\varphi,I;t)\right|=\left|\sum_{0\neq \gamma\in \mathbb{Z}^{d} }\sum_{0<\alpha_{1}\leq\alpha}\binom{\alpha}{\alpha_{1}}D_{I}^{\alpha_{1}}(\langle\omega(I),
\gamma\rangle^{-1})D_{I}^{\alpha-\alpha_{1}}f_{\gamma}(I)\gamma^{\beta}e^{{\rm i}\langle \gamma,\varphi\rangle}\right|\\
& \leq \sum_{0\neq \gamma\in \mathbb{Z}^{d}}\gamma^{\beta}W(|\gamma|)^{-1}e^{-C^{-1}|\gamma|^{\frac{1}{\sigma}}}\sum_{0<\alpha_{1}\leq\alpha}\varepsilon C_{1}^{|\alpha-\alpha_{1}|}C_{0}^{|\alpha_{1}|+1}\frac{\alpha!}{(\alpha-\alpha_{1})!}(\alpha-\alpha_{1})!^{\mu}\\
&~~\times \max_{0\leq j\leq|\alpha_{1}|}\Big|(|\alpha_{1}|-j)!^{\rho}|\gamma|^{j-|\beta|}\Delta^{j-|\beta|}(|\gamma|)\Big|\\
& \leq \varepsilon C_{0}\sum_{0\neq \gamma\in \mathbb{Z}^{d}}\gamma^{\beta}\Delta(|\gamma|)e^{-C^{-1}|\gamma|^{\frac{1}{\sigma}}}\sum_{0<\alpha_{1}\leq\alpha}\frac{\alpha!}{(\alpha-\alpha_{1})!}C_{1}^{|\alpha-\alpha_{1}|}
C_{0}^{|\alpha_{1}|} (\alpha-\alpha_{1})!^{\mu}\beta!^{\rho} \\
&\leq \varepsilon C_{0}\sum_{m=1}^{\infty}\sum_{|\gamma|=m}|\gamma|^{|\beta|}\Delta(|\gamma|)e^{-C^{-1}|\gamma|^{\frac{1}{\sigma}}}
\sum_{0<\alpha_{1}\leq\alpha}C_{1}^{|\alpha-\alpha_{1}|}C_{0}^{|\alpha_{1}|}\frac{\alpha!}{(\alpha-\alpha_{1})!}
(\alpha-\alpha_{1})!^{\mu}\beta!^{\rho}\\
&\leq 2^{d}\varepsilon C_{0}\sum_{m=1}^{\infty}m^{|\beta|+d-1}\Delta(m)e^{-C^{-1} m^{\frac{1}{\sigma}}}\sum_{0<\alpha_{1}\leq\alpha}(C_{1}^{-1}C_{2}^{-1}C_{0})^{|\alpha_{1}|}C_{1}^{|\alpha|}C_{2}^{|\beta|}\alpha!^{\mu}\beta!^{\rho}\\
&\leq 2^{d}\varepsilon C_{0}\sup_{t\geq0}t^{|\beta|+d+1}\Delta(t)e^{-C^{-1} t^{\frac{1}{\sigma}}}\sum_{m=1}^{\infty}\frac{1}{m^{2}}
C_{1}^{|\alpha|}C_{2}^{|\beta|}\alpha!^{\mu}\beta!^{\rho}\\
& \leq \frac{\pi^{2}}{6}2^{d}C_{0}\Gamma_{s}(\eta)\varepsilon C_{1}^{|\alpha|}C_{2}^{|\beta|}\alpha!^{\mu}\beta!^{\rho}\\
& \leq \frac{\pi^{2}}{6}2^{d}C_{0}e^{\eta T^{\frac{1}{\sigma}}}\varepsilon C_{1}^{|\alpha|}C_{2}^{|\beta|}\alpha!^{\mu}\beta!^{\rho}
\end{align*}
with $s=|\beta|+d+1$ and $\eta=C^{-1}>0$ through the use of Lemma \ref{a7}.
It is important to note that
\begin{align*}
\sum_{m=1}^{\infty}\frac{1}{m^{2}}=\frac{\pi^{2}}{6}<\infty.
\end{align*}
Let $\epsilon=C_{1}^{-1}C_{2}^{-1}C_{0}$, we thus have
\begin{align*}
\sum_{0<\alpha_{1}}(C_{1}^{-1}C_{2}^{-1}C_{0})^{|\alpha_{1}|}<\sum_{0<\alpha_{1}}\epsilon^{|\alpha_{1}|}=\epsilon\sum_{s=2}^{\infty}s^{d}\epsilon^{s-2}<1.
\end{align*}
We get the final estimate
\begin{align*}
\left|D_{I}^{\alpha}D_{\varphi}^{\beta}\psi(\varphi,I)\right|\leq \varepsilon D C_{1}^{|\alpha|}C_{2}^{|\beta|}\alpha!^{\mu}\beta!^{\rho}
\end{align*}
by taking $D=\frac{\pi^{2}}{6}2^{d}C_{0}e^{\eta T^{\frac{1}{\sigma}}}$, where $\eta$ is a linear function of $s$.

\subsection{Gevrey function}
In this section, we give the definition of the anisotropic classes, which are classes of Gevrey functions with differing regularity in individual variables.
\begin{definition}\label{at}
Given $\sigma,\mu\geq1$ and $L_{1},L_{2}>0$, we say that $f\in \mathcal{G}^{\sigma,\ \mu}(\mathbb{T}^{d}\times D)$, if $f\in C^{\infty}(\mathbb{T}^{d}\times D)$ with associated with norm:
\begin{align}\label{l1}
\|f\|_{L_{1},L_{2}}=:\sup_{\alpha,\beta}\sup_{(\varphi,I)}|\partial_{\varphi}^{\alpha}\partial_{I}^{\beta}f(\varphi,I)|L_{1}^{-\alpha}L_{2}^{-\beta}\alpha!^{-\sigma}\beta!^{\mu}
<\infty.
\end{align}
\end{definition}
From Taylor's Theorem, we find that Gverey functions have an important property, that is, if a Gevrey function has vanishing derivatives, then locally it is super-exponentially small.
\begin{proposition}\label{as}
Let $X\subset \mathbb{R}^{d}$ open. For each Gevrey function $f\in \mathcal{G}^{\mu}(X),\mu>1$, there exit $\eta>0,\ c>0$ and $C_{1}>0$ depending only on the constant $C=C(X)$ such that
\begin{align*}
f(I_{0}+r)=\sum_{|\alpha|\leq\eta|r|^{1/(1-\mu)}}f(I_{0})r^{\alpha}+R(I_{0},r),
\end{align*}
where $f_{\alpha}(I_{0})=\partial^{\alpha}f(I_{0})/\alpha!$ and
\begin{align*}
|\partial_{I}^{\alpha}R(I_{0},r)|\leq C_{1}^{|\beta|+1}\beta!^{\mu}e^{-c|r|^{-\frac{1}{\mu-1}}},0<|r|\leq r_{0}.
\end{align*}
\end{proposition}
In anisotropic Gevrey class, we have the following implicit function theorem due to Komatsu \cite{MR531445}. Fix $\rho\geq1$. Let $F=(f_{1},\cdots,f_{d})$ be a Gevrey function of the class $\mathcal{G}^{\rho,\rho+1}(X\times\mathbb{T}^{d},\mathbb{R}^{d})$. We assume that there are constants $A_{0}>0$ and $L_{1},L_{2}>0$, and a small parameter $0<\varepsilon\leq 1$, such that
\begin{align*}
|D_{x}^{\alpha}D_{\theta}^{\beta}(F(x,\theta)-x)|\leq \varepsilon A_{0}L_{1}^{|\alpha|}L_{2}^{|\beta|}\alpha!^{\rho}\beta!^{\rho+1}
\end{align*}
for any multi-indices $\alpha,\beta$.
\begin{proposition}\label{s1}
Suppose that $F\in \mathcal{G}^{\rho,\rho+1}_{L_{1},L_{2}}(X\times\mathbb{T}^{d},\mathbb{R}^{d})$, where $X\subset\mathbb{R}^{d}$ and $L_{1}\|F(x,\theta)-x\|_{L_{1},L_{2}}\leq 1/2$. Then there exists a local solution $x=g(y,\theta)$ to the implicit equation
\begin{align*}
F(x,\theta)=y,\ y\in Y,
\end{align*}
where $Y$ is a open subset of Euclidean spaces. Moreover, there exist constants $A,C$ dependent only on $\rho,n,m$ such that $g\in\mathcal{G}^{\rho,\rho+1}_{CL_{1},CL_{2}}(Y\times\mathbb{T}^{d},X)$ with $\|g\|_{CL_{1},CL_{2}}\leq A\|F\|_{L_{1},L_{2}}$.
\end{proposition}

\subsection{h-PDOs, h-FIOs, Gevrey class symbol}
To enhance the clarity and appeal of the content that follows, let's introduce some concepts. Firstly, we define the class of Gevrey symbols that we need. Let $D_{0}$ be a bounded domain in $\mathbb{R}^{d}$. Fix $\sigma,\mu>1,\bar{\rho}\geq\mu+\sigma-1$, and set $l=(\sigma,\mu,\bar{\rho})$. We then introduce a class of formal  Gevrey symbols, $F\mathcal{S}_{l}(\mathbb{T}^{d}\times D_{0})$, in the following discussion.
Consider a sequence of smooth functions $p_{j}\in C_{0}^{\infty}(\mathbb{T}^{d}\times D_{0}),\ j\in\mathbb{Z}_{+}$ such that $p_{j}$ is supported within a fixed compact subset of $\mathbb{T}^{d}\times D_{0}$.
\begin{definition}
We say that
\begin{align}\label{a1}
\sum_{j=0}^{\infty}p_{j}(\varphi,I)h^{j}
\end{align}
is a formal symbol in $F\mathcal{S}_{l}(\mathbb{T}^{d}\times D_{0})$, if there exits a positive constant $C$ such that $p_{j}$ satisfies the estimate
\begin{align*}
\sup_{\mathbb{T}^{d}\times D_{0}}|\partial_{\varphi}^{\beta}\partial_{I}^{\alpha}p_{j}(\varphi,I)|\leq C^{j+|\alpha|+|\beta|+1}\beta!^{\sigma}\alpha!^{\mu}j!^{\bar{\rho}}
\end{align*}
for any $\alpha,\beta$ and $j$.
\end{definition}
\begin{definition}
The function $p(\varphi,I;h),(\varphi,I)\in\mathbb{T}^{d}\times \mathbb{R}^{d}$ is called a realization of the formal symbol \eqref{a1} in $\mathbb{T}^{d}\times D_{0}$, if for each $0<h\leq h_{0}$, it is smooth with respect to $(\varphi,I)$ and has compact support in $\mathbb{T}^{d}\times D_{0}$, and if there exits a positive constant $C_{1}$ such that
\begin{align*}
\sup_{\mathbb{T}^{d}\times D_{0}\times(0,h_{0}]}|\partial_{\varphi}^{\beta}\partial_{I}^{\alpha}p(\varphi,I;h)-\sum_{j=0}^{N}p_{j}(\varphi,I)h^{j}|
\leq h^{N+1}C^{N+|\alpha|+|\beta|+2}\beta!^{\sigma}\alpha!^{\mu}(N+1)!^{\bar{\rho}}
\end{align*}
for any multi-indices $\alpha,\beta$ and $N\in\mathbb{Z}_{+}$.
\end{definition}
\begin{remark}
When $\sigma=\mu$ and $\bar{\rho}=2\sigma-1$, we set $\mathcal{S}_{l}=\mathcal{S}^{\sigma}$.
\end{remark}
Given two symbols $p,q\in\mathcal{S}_{l}(\mathbb{T}^{d}\times D_{0})$, we denote their composition by $p\circ q\in\mathcal{S}_{l}(\mathbb{T}^{d}\times D_{0})$, which is the realization of
\begin{align*}
\sum_{j=0}^{\infty}c_{j}h^{j}\in F\mathcal{S}_{l}(\mathbb{T}^{d}\times D_{0}),
\end{align*}
where
\begin{align*}
c_{j}(\varphi,I)=\sum_{j=r+s+|\gamma|}\frac{1}{\gamma!}D_{I}^{\gamma}p_{r}(\varphi,I)\partial_{I}^{\gamma}q_{s}(\varphi,I).
\end{align*}
We can associate an $h-$pseudodifferential operator for each symbol $p$ by:
To each symbol $p\in \mathcal{S}_{l}(\mathbb{T}^{d}\times D_{0})$,
\begin{align}\label{bv}
P_{h}u(x)=(2\pi h)^{-d}\int_{\mathbb{T}^{d}\times \mathbb{R}^{d}}e^{i(x-y)\cdot\xi/h}p(x,\xi;h)u(y)d\xi dy,
\end{align}
for $u\in C_{0}^{\infty}(\mathbb{T}^{d})$.

\begin{definition}\label{bw}
We define the residual class of symbol, $\mathcal{S}_{l}^{-\infty}$, as the collection of realizations of the zero formal symbol.
\end{definition}
To help everyone's understanding of this definition, we provide the following example:
if $f-g\in \mathcal{S}_{l}^{-\infty}$, we say that $f$ is equivalent to $g$ denoted by $f\sim g$. Gevrey symbols are precisely the equivalence classes of $\sim$.
\begin{remark}
The exponential decay of residual symbols is a key strengthening that comes from working in a Gevrey symbol class, as opposed to the standard Kohn-Nirenberg classes (see \cite{MR1996773}).
\end{remark}
From \cite{MR0388463}, we have:
\begin{definition}\label{bu}
Let $X,Y$ be open sets in $\mathbb{R}^{d}$ and $\mathbb{R}^{m}$, and let $\phi$ be a real valued function of $(x,y,\xi)\in X\times Y\times \mathbb{R}^{N}$ which is positively homogeneous of degree 1 with respect to $\xi$ and infinitely differentiable for $\xi\neq0$. We define the operator as follows
\begin{align*}
F_{h}u(x)=(2\pi h)^{-n}\int_{\mathbb{T}^{n}\times \mathbb{R}^{n}}e^{i\phi(x,y,\xi)/h}f(x,y,\xi;h)u(y)d\xi dy
\end{align*}
with a symbol $f\in \mathcal{S}_{l}(X\times Y\times \mathbb{R}^{N})$. Then, we refer to the operator $F_{h}$ as $h-$Fourier integral operator (h-FIOs).
\end{definition}
We observe that if we identify the phase function $\phi(x,y,\xi;t)$ as $\langle x-y,\xi\rangle$, then $F_{h}$  is a pseudodifferential operator. In other words, the class of Fourier integral operators encompasses a broader category than the class of pseudodifferential operators. Readers interested in Fourier integral operator can refer to the references \cite{MR0388463}-\cite{MR4248008} and so on.

\subsection{Measure-theoretic lemma}
We recall a measure-theoretic lemma in \cite{MR4404789}.
\begin{lemma}\label{an}
If $A\subset(0,t_{0})$ with
\begin{align*}
\frac{meas(A\cap(x-r,x+r))}{meas((0,t_{0})\cap(x-r,x+r)}<\epsilon
\end{align*}
for all $x\in(0,t_{0})$, then $meas(A)<2\epsilon t_{0}$.
\end{lemma}
\begin{proof}
We have
\begin{align*}
\frac{1}{meas((x-r,x+r)\cap(0,t_{0}))}\int_{(x-r,x+r)\cap(0,t_{0})}1_{A}(y)dy<\epsilon
\end{align*}
for each $x\in(0, t_{0})$. Integrating in $x$, we obtain
\begin{align*}
\int_{0}^{t_{0}}\int_{(x-r,x+r)\cap(0,t_{0})}1_{A}(y)dydx<2r\epsilon t_{0}.
\end{align*}
Fubini$^{\prime}$s theorem then implies
\begin{align*}
\int_{A}meas((x-r,x+r)\cap(0,t_{0}))<2r\epsilon t_{0},
\end{align*}
and $meas(A)<2r\epsilon t_{0}$ follows.
\end{proof}




\section*{Acknowledgments}
 The second author (Y. Li) was supported by  National Natural Science Foundation of China (12071175 and 12471183).

\end{document}